\documentclass[aps, prx,superscriptaddress, twocolumn, longbibliography]{revtex4-1}%
\usepackage{amsfonts, cancel, xcolor, mathtools, amsmath, amssymb, graphicx, mathrsfs, physics, bbm, times, subfigure}
\definecolor{myurlcolor}{rgb}{0,0,0.7}
\definecolor{myurlcolor1}{rgb}{0,0.7,0.1}
\definecolor{myrefcolor}{rgb}{0,0,0.7}
\usepackage[unicode=true,pdfusetitle, bookmarks=true,bookmarksnumbered=false,bookmarksopen=false, breaklinks=false,pdfborder={0 0 0},backref=false,colorlinks=true, linkcolor=myrefcolor,citecolor=myurlcolor1,urlcolor=myurlcolor]{hyperref}
\definecolor{mybox}{rgb}{0.5,0.5,0.3}
\providecommand{\U}[1]{\protect\rule{.1in}{.1in}}
\newtheorem{theorem}{Theorem}
\newtheorem{claim}{Claim}
\newtheorem{corollary}{Corollary}
\newtheorem{definition}{Definition}

\newtheorem{lemma}{Lemma}
\newtheorem{proposition}{Proposition}

\newenvironment{proof}[1][Proof]{\noindent\textbf{#1.} }{\ \rule{0.5em}{0.5em}}

\newcommand{\wt}[1]{\widetilde{#1}}
\newcommand{\mch}{\mathcal{H}_S}
\newcommand{\mh}{\hat{H}_S}

\newcommand{\psh}{\mathsf{P}(S)}
\newcommand{\Op}[1]{\ket{#1}\bra{#1}}
\newcommand{\pmax}{\mathrm{P}_{\max}(\rho)}
\newcommand{\psuc}{p_{\mathrm{succ}}\left(\mathcal{T},\rho,\{M_b\}\right)}
\newcommand{\psucc}{p_{\mathrm{succ}}\left(\mathcal{T},\rho\right)}
\newcommand{\psuccpassive}{p_{\mathrm{succ}}^{\mathrm{P}}\left(\mathcal{T}\right)}

\newcommand{\mc}[1]{\mathcal{#1}}
\newcommand{\msc}[1]{\mathscr{#1}}

\newcommand{\roa}[1]{\mathsf{A}_r\left(#1\right)}
\newcommand{\troa}[1]{\mathsf{A}'_r\left(#1\right)}

\newcommand{\aw}{\mathsf{A}_w\left(\rho\right)}
\newcommand{\taw}{\mathsf{A}'_w\left(\rho\right)}
\newcommand{\awe}[1]{\mathsf{A}_w\left(#1\right)}
\newcommand{\atryp}{\mathsf{A}'_{\mathrm{try}}\left(\rho,H,n\right)}
\newcommand{\atry}{\mathsf{A}_{\mathrm{try}}\left(\rho,H\right)}

\begin{document}

\title{A  resource theory of activity for quantum thermodynamics in the absence of heat baths}

\author{Swati}
\email{swati@cs.hku.hk}
\affiliation{QICI Quantum Information and Computation Initiative, Department of Computer Science,\\
The University of Hong Kong, Pokfulam Road, Hong Kong}
\affiliation{Institute for Quantum Science and Engineering, Department of Physics,\\
Southern University of Science and Technology (SUSTech), Shenzhen 518055, China}
\author{Uttam Singh}
\email{uttam@cft.edu.pl}
\affiliation{Center for Theoretical Physics, Polish Academy of Sciences,\\ Aleja Lotnik\'ow 32/46, 02-668 Warsaw, Poland}
\affiliation{Centre of Quantum Science and Technology, International Institute of Information Technology,\\
Hyderabad, Gachibowli, Telangana 500032, India}

\author{Giulio Chiribella}
\email{giulio@cs.hku.hk}
\affiliation{QICI Quantum Information and Computation Initiative, Department of Computer Science,\\
The University of Hong Kong, Pokfulam Road, Hong Kong}
\affiliation{Department of Computer Science, University of Oxford, Wolfson Building, Parks Road, Oxford, UK}
 \affiliation{HKU-Oxford Joint Laboratory for Quantum Information and Computation}
\affiliation{Perimeter Institute for Theoretical Physics, 31 Caroline Street North, Waterloo, Ontario, Canada}

\begin{abstract}
Active states, from which work can be extracted by time-dependent perturbations, are an important resource for quantum thermodynamics in the absence of heat baths.
Here we characterize this resource, establishing a resource theory that captures the operational scenario where an experimenter manipulates a quantum system by means of energy-preserving operations and resets to non-active states. Our resource theory comes with simple conditions for state convertibility and an experimentally accessible  resource quantifier that determines the maximum advantage of active states in the task of producing approximations of the maximally coherent state by means of energy-preserving quantum operations. 
\end{abstract}

\maketitle

{\em Introduction.} Quantum technologies have brought new ways to address and control individual quantum systems,  motivating  an extension of the framework of thermodynamics  to the quantum domain \cite{Geusic1967, Faucheux1995, Howard1997, Rousselet1994, Scully2002, Hanggi2009,Vinjanampathy2016, Goold2016, Lostaglio2019, Binder2019}.  A successful approach to quantum thermodynamics is provided by  quantum resource theories \cite{Chitambar2019,coecke2016mathematical}.  Quantum resource theories are  built on the notions   of free states and free operations, that is,  states and operations that are not regarded as  resources. An important example is the resource theory of athermality  \cite{Michal2013, Brandao2013, Faist2015,Brandao2015b}, in which the free states are the Gibbs states, and the free operations are the thermal operations, that is, operations achievable by letting the system interact with a heat bath at a given temperature through an energy-preserving  unitary evolution \cite{Brandao2013}.   

In the absence of heat baths,  the development of a resource-theoretic framework for quantum thermodynamics is a challenging problem. A clue comes from considerations about the maximum work extractable by subjecting the system's Hamiltonian to time-dependent perturbations \cite{Pusz1978, Lenard78}  (see also \cite{Marcin2021}).  
 In this setting, the perturbed evolution generally changes the  system's energy, and the energy difference  can in principle be converted into work.  The maximum energy difference achievable for a given state $\rho$, hereafter denoted by ${\sf Erg} (\rho)$,  is known as the   {\em ergotropy} \cite{Allahverdyan2004}   and is given by
 \begin{align}
\label{ergotropy}
{\sf Erg} (\rho)=  \langle H  \rangle_\rho    -   \min_{U}   \langle H \rangle_{U\rho U^\dag}  \,,  
\end{align}
where  $\langle H \rangle_\sigma  = \Tr [  H\,  \sigma]$ denotes the expectation value of the system's Hamiltonian $H$ on a state $\sigma$, and the minimization runs over arbitrary unitary operators $U$.   
 The states with zero ergotropy  are called {\em passive}, while the other states are  called {\em active}.     These sets of states  has been  studied extensively  throughout the development of quantum thermodynamics, see e.g. Refs. \cite{Allahverdyan2004, Skrzypczyk2015, Perarnau2015, Perarnau-Llobet2015, Avijit2016,  Brown2016,  Amit2016, Sparaciari2017, Uttam2019, Perarnau-Llobet2019, Mir2019, Mir2020, Salvia2020, Raffaele2020, Salvatore2021, Kornikar2021,  Nikolao2021, Tamal2022,  Samgeeth2022}. 
 

Since active states are an important thermodynamical resource,  a natural first step is to develop a resource theory of activity. 
In this resource theory,  the passive states are the obvious choice of free states.   
The choice of free operations, however, is far from obvious.  A minimal requirement  is that free operations should  transform passive states into passive states, {\em i.e.} they should be {\em passivity-preserving} \cite{Uttam2021}.  However, the set of  passivity-preserving operations is difficult to characterize and, more importantly, is not consistent with notion of ergotropy as a resource: as we will show later in this paper, there exist passivity-preserving operations that increase the ergotropy. 
 A fix to this problem is to restrict the attention to the smaller set of {\em ergotropy non-increasing  operations}.     Unfortunately, however, also  the ergotropy non-increasing operations  are  difficult to characterize.  
Moreover, the condition of ergotropy non-increase depends on  the  eigenvalues of the Hamiltonian  in a fine-tuned way which does not seem to correspond to any operationally relevant scenario. 

In this paper, we formulate a resource theory of activity where the free operations can be implemented through   {\em energy-preserving channels and passive resets  (EPCPR)}, the latter being the operations that reset the system to  passive states.    We focus on the setting where  the Hamiltonian is non-degenerate, showing that in this case  all EPCPR operations are ergotropy non-increasing.  We also provide necessary and sufficient  conditions for the convertibility of  relevant sets of quantum states under EPCPR operations,  and we  define a broad class of resource quantifiers, including several variants of the relative entropy distance from the set of passive states. The latter represent an analogue of the second law of quantum thermodynamics,   in a similar way as relative entropy distances to the Gibbs state have been used as a quantum second laws in the presence of heat baths \cite{Brandao2015b}.    Among the entropic quantifiers, we consider   in particular the max relative entropy of activity, which we show to be experimentally accessible through the measurement of a suitable observable  witnessing the presence of activity.    Finally, we show that the max relative entropy of activity admits an operational interpretation as the maximum advantage offered by active states in the task of generating approximations of the maximally coherent state by means of energy-preserving quantum operations.   Overall, our resource theory of activity   provides basic tools for the study of single-system quantum thermodynamics in the absence of heat baths, and a solid framework for connecting activity with other quantum resources, such as quantum coherence and quantum entanglement.

{\em Resource theory of activity.} Quantum resource theories \cite{Chitambar2019,coecke2016mathematical} play a central role in quantum information, see e.g. Refs. \cite{Baumgratz2014, Brandao2015d,Streltsov2015, Singh2015, Bu2016,  Dana2017, Anshu2018, Albarelli2018, Anand2019, Wang2019, Wang2019b, Liu2020, Kristjansson2020, Bhattacharya2021, Berk2021,  Wu2021, Swati2022, Roy2022, Milz2022, Halpern2022, Patra2022}.  A quantum resource theory is defined  by  a set of free states and a set of free operations, usually motivated by physical considerations or by  mathematical convenience.  In the following  we  formulate a resource theory of activity whose  free operations are both physically motivated and mathematically tractable.

Let $S$ be   a $d$-dimensional quantum system with Hamiltonian  $H$, and let 
 ${\sf{St}} (S)$ be the state space of  system $S$.    In this paper we take the Hamiltonian to be non-degenerate, meaning that $H$ has $d$ distinct eigenvalues  $\{E_i\}_{i=1}^d$, corresponding to   a basis of eigenstates $(|i\rangle)_{i=1}^d$ satisfying the condition $H  |i\rangle  =  E_i\,  |i\rangle$.   A state $\tau \in {\sf St} (S)$ is called passive if its ergotropy is zero.  Passive states have a simple characterization \cite{Lenard78}: the state $\tau$    is passive   if and only if it  is of the form 
$\tau  =  \sum_{i=1}^{d}  \,    p_{i} \,  |i\rangle\langle i|$,  
where  $(p_i)_{i=1}^{d}$ is a probability distribution satisfying the condition
$p_i  \le  p_j$ for every $i$ and $j$ such that $\Delta E_{ij}  :  =E_i - E_j \ge 0$.
In the following we will always assume that the eigenvalues of the Hamiltonian have been listed in increasing order, as $E_1  <E_2  < \dots  < E_d$. 


The set $\psh$ of passive states  is the set of free states in our resource theory.   Since in general the tensor product of two passive states is  not a passive state \cite{Lenard78, Sparaciari2017}, our resource theory has to be regarded as  a single-system resource theory, in the sense that it does not consider parallel  composition  of resources.   Nevertheless, one can always consider the situation where the system $S$ is multipartite, and ergotropy  is defined  with respect  to the set of all global unitary operations performed on the composite system.  
The fact that the tensor product of two free states may not be free corresponds to the fact that  local unitary operations may not be able to extract any amount of work even if the global state is active  \cite{Brown2016, Perarnau-Llobet2019, Salvia2020}.

Let us now specify the free operations.  A basic requirement in  every resource theory is that the free operations should not transform free states into non-free states \cite{Chitambar2019}. In a resource theory of activity,  the condition is that free operations must be passivity-preserving \cite{Uttam2021}, that is, they should map $\psh$ into itself.   
General passivity-preserving operations, however, are not consistent with the notion of ergotropy as a resource: for example, consider the quantum channel $\mc C$ defined by  
\begin{align}
\nonumber
&\mc C( \rho )  =  \sum_{k=1}^4  \,   \langle k|   \rho  |k\rangle   \,    \gamma_k    \\ 
\nonumber 
&\gamma_1 =|1\rangle \langle 1|   \qquad  \gamma_2 =  \frac 13 \,  |2\rangle \langle 2|  +  \frac 23  |1\rangle \langle 1| \\
&  \gamma_3 =  \frac 23 \,  |2\rangle \langle 2|  +  \frac 13  |1\rangle \langle 1|  \qquad \gamma_4 =|2\rangle \langle 2|    \, .  \label{counterexample}
\end{align}
It is easy to check that this channel maps all passive states into passive states.  On the other hand, it  changes the ergotropy of  the state  $\rho  =    (|1\rangle \langle 1|   +  |3\rangle \langle 3|   +  |4\rangle \langle 4| )/3$ from    ${\sf Erg}(\rho)   =  \Delta E_{42}/3$ to  ${\sf Erg} (\mc C  (\rho))  =  \Delta E_{21}/9$, which is larger than ${\sf Erg}(\rho)$    whenever $\Delta E_{21}  >  3\,  \Delta E_{42}$. 


A more fitting set of operations for a resource theory of activity is the set of ergotropy non-increasing channels, that is, quantum channels $\mc C$ satisfying the condition $  {\sf Erg}  (\mc C  (\rho))  \le {\sf Erg}(\rho)$  for every state $\rho \in {\sf St} (S)$.   Note that every ergotropy non-increasing channel is automatically passivity-preserving. A limitation of the  ergotropy non-increasing channels, however, is that they are  difficult to characterize.  Moreover, the condition of ergotropy non-increase amounts to a rather fine-tuned relation between the channel $\mc C$ and the eigenvalues of the Hamiltonian, which does not seem to correspond to any operationally relevant scenario.

We now introduce the set of quantum channels  achievable by energy-preserving  channels and passive resets (EPCPR).   
An EPCPR channel is a  quantum channel  $\mc C$  of the form 
\begin{align}
\mc C  (\rho)    =  p  \,  \mc E   (\rho)  + (1-p)  \,  \tau \, , 
\end{align}
where $p \in[0,1]$ is a probability, $\tau$ is an arbitrary passive state, and $\mc E$ is an arbitrary   energy-preserving channel  \cite{chiribella2017optimal}, that is, a quantum channel that preserves the probability distribution of the energy for every possible state; in formula,    $ \langle i|       \mc E (\rho) |i\rangle  = \langle i|   \rho | i\rangle$ for every state $\rho$ and for every  $i\in  \{  1,\dots,  d\}$. Operationally, EPCPR channels  can be implemented by randomly choosing whether to reset the system to a given passive state  or to perform an energy-preserving channel.  In turn, energy-preserving channels can be operationally characterized as the quantum channels that can be  implemented by setting up an energy-preserving interaction between the system and  an auxiliary system with fully degenerate Hamiltonian \cite{chiribella2017optimal}.     Notably, the dimension of the auxiliary system can be chosen without loss of generality to be equal to the dimension of  system $S$.

  
As anticipated, every EPCPR channel is ergotropy non-increasing. The proof is as follows: first, the definition of ergotropy (\ref{ergotropy}) implies the inequality ${\sf Erg}  (\mc C (\rho))  \le      p  \,  {\sf Erg}  (\mc E  (\rho)) + (1-p)\,   {\sf Erg}  (\tau)=   p\,  {\sf Erg}  (\mc E  (\rho))  \le  {\sf Erg}  (\mc E  (\rho)).$  On the other hand, all energy-preserving channels are unital  \cite{chiribella2017optimal}, namely $\mc E (I)  =  I$.        Since the state  $\rho$ can be converted into the state $\mc E (\rho) $ by means of a unital channel, it can also be converted into it by means of random unitary operations \cite{gour2015resource,chiribella2017microcanonical}, say $\mc E(\rho)  =  \sum_j \,  q_j\,  \mc U_j  (\rho)$ for some probabilities $q_j$ and some unitary channels $\mc U_j$. Hence, we have ${\sf Erg}  (  \mc E  (\rho))  =  \Tr[ H \mc E(\rho) ]-  \min_{\mc U} \Tr[ H \mc U \mc E(\rho) ]  =  \Tr[  H  \,  \rho]  - \min_{\mc U}  \sum_j q_j  \Tr[ H\,  \mc U \mc U_j(\rho) ]  \le  \Tr[  H  \,  \rho]  -  \sum_j q_j  \min_{\mc U}   \Tr[ H\,  \mc U \mc U_j(\rho) ]    =  {\sf Erg} (\rho)$.  

We note in passing that the above proof holds also if the set of energy-preserving channels  is replaced by the larger set of  {\rm unital, average energy non-increasing  channels (UAENIC)}, namely unital channels  $\mc E$ satisfying the condition $\langle H\rangle_{\mc E   (\rho)}  \le   \langle H  \rangle_\rho\, , \forall \rho  \in  {\sf St}  (S)$. By the above proof, the set of operations generated by random mixtures of UAENIC operations and passive resetting (UAENICPR) are all ergotropy non-increasing.  

{\em Convertibility conditions.}    A key question in any resource theory is when a given initial  $\rho$ can be converted into a final state $\rho'$ via free operations.  In the resource theory of activity, the answer to this question has an easy-to-test form  in several   relevant cases.  The simplest case is when the input state is diagonal in the energy eigenbasis.  In this case, the convertibility is controlled by two numbers, $p_+$ and $p_-$, defined as follows 
\begin{align}
p_+   :=   \min\left .\left \{  \frac{ \Delta'_{m,m+1} }{\Delta _{m,m+1}} ~\right|~  1\le m \le d-1  \, ,  \Delta_{m,m+1}  \ge 0\right\}   \\
p_-   :=   \max\left .\left \{  \frac{ \Delta'_{m,m+1} }{\Delta _{m,m+1}} ~\right|~  1\le m \le d-1  \, ,  \Delta_{m,m+1}  < 0\right\}  \, ,   
\end{align}
with $\Delta_{m,n}  :  =  \rho_{mm}   -  \rho_{nn}$ and $\Delta'_{m,n}  :  =  \rho'_{mm}   -  \rho'_{nn}$ (here we adopt the convention $\Delta'_{m,n}/\Delta_{m,n}  =  {\sf sign}  (\Delta'_{m,n}) \times \infty$ when $\Delta_{m,n}  =  0$.)   
  With this notation, we can state a necessary and sufficient condition for convertibility:  the state transition $\rho  \mapsto  \rho'$ is achievable by EPCPR channels if and only if $\rho'$ is diagonal in the energy eigenbasis,  $p_{-}\le p_+$, and  $[p_-,  p_+]  \cap  [0,1]\not   =  \emptyset$.

Another relevant case is that of input states  that have coherence across all energy levels, namely  $\rho_{mn}  \not =  0$ for every $m$ and $n$.  In this case, the state transition $\rho \mapsto  \rho'$ is achievable by EPCPR if and only if $p_{-}\le p_+$,    $[p_-,  p_+]  \cup  [0,1]\not   =  \emptyset$,  and  $\min{\rm eigv}  (R)  \ge  -  \min\{  p_+,  1\}$, where $\min{\rm eigv}  (R)$ is the minimum eigenvalue of the Hermitian matrix defined by $R_{mn}  :  =  \rho'_{mn}/\rho_{mn}$.   The derivations of these conditions is  provided in Appendix \ref{append:state-transformation}.

{\em Activity measures.} We have seen that  ergotropy is a resource monotone in our resource theory of activity.  Another class of examples are geometric quantities associated to the set of passive states.  In particular, here we define  the {\em $\alpha$-relative entropies of activity}   $R^{\rm act}_\alpha  (\rho)   :=  \min_{\tau \in  \psh}  \widetilde D_\alpha   \big(\rho  \|  \tau  \big) $, where $\widetilde D_\alpha   \big(\rho  \|  \tau  \big)  $ is the  sandwiched R\'enyi relative entropy  \cite{Muller-Lennert2013, Beigi2013, Wilde2014}, defined as $\widetilde D_\alpha   \big(\rho  \|  \tau  \big)  : =   \Tr \left[   \left(  \tau^{\frac{1-\alpha}{2\alpha}}  \rho  \tau^{\frac{1-\alpha}{2\alpha}}\right)^{\alpha} \right]/(\alpha-  1)$ when ${\sf Supp}  (\rho)  \subseteq {\sf Supp (\tau)}$ and $\infty$ otherwise.     The definition and the data processing inequality for the sandwiched R\'enyi entropies imply that   $R^{\rm act}_\alpha$ is monotonically non-increasing under arbitrary passivity-preserving operations for every $\alpha  \ge 1/2$.   Hence, it is a resource monotone for the resource theory of activity.  Two important limiting cases  are  $\alpha  \to 1$ and $\alpha \to \infty$, in which case $R^{\rm act}_\alpha$ becomes the {\em relative entropy of activity} $R^{\rm act} (\rho)   :  =  \min_{\tau  \in  \psh}  \Tr [  \rho   (\log \rho  -  \log \tau)]$  and the {\em max relative entropy of activity} $R_{\max}^{\rm act}  (\rho)  :  =  \min_{\tau  \in  \psh}  D_{\max}  (\rho \|  \tau)$, with $D_{\max}  (A\|  B)   =  \min \{s\geq 0: A \leq 2^s B\}$ for arbitrary positive operators $A$ and $B$.   

Similar monotones are the {\em inverse $\alpha$-relative entropies of activity}  $  \cancel{R}_\alpha^{\rm act}  (\rho)   :=  \min_{\tau \in  \psh}  \widetilde D_\alpha   \big(  \tau  \|  \rho  \big) $, which are also non-increasing under arbitrary passivity-preserving operations for every $\alpha \ge 1/2$.    The non-increase of    $R^{\rm act}_\alpha$ and  $\cancel{R}_\alpha^{\rm act}$ under energy-preserving channels and passive resets can be seen as an analogue of the second law of thermodynamics  in the absence of heat baths, similarly to the quantum second laws put forward in Ref. \cite{Brandao2015b}.   

In the following we will focus on the max relative entropies  $R^{\rm act}_{\max}  (\rho)$ and $\cancel R_{\max}^{\rm act}  (\rho)  :  =  \min_{\tau \in\psh}  D_{\max}  (\tau\| \rho)$, which play a role in one-shot tasks.    Here,   $R_{\max}^{\rm act}  (\rho)$ is related to  the robustness of  activity with respect to randomizations: indeed, one has  the relation   $R^{\rm act}_{\max}   (\rho)  =    \log  [   \roa \rho + 1   ]$, where  $\roa \rho  := \min_{\sigma\in {\sf St  (S)}}\left\{t\geq 0 ~|~  (\rho+t\sigma)/(1+t)  \in \psh\right\}$ quantifies the minimum amount of randomization needed to turn $\rho$ into a passive state.   Following similar terminology in general resource theories (see e.g. \cite{Takagi2019b}), we call $\roa \rho$ the {\em (generalized) robustness of activity} of the state $\rho$.     Similarly, $\cancel R_{\max}^{\rm act} (\rho)$ is related to the maximum weight of a passive state in a convex decomposition of  $\rho$: indeed, we have the relation   $\cancel R_{\max}^{\rm act}  (\rho)   =   - \log [ 1-  \aw  ]$, with $\aw :=\min_{\sigma\in\sf{St}  (S) ,~ \tau\in\psh}\{t\geq 0:\rho= t\sigma+(1-t)\tau\}$. We call $\aw$ the {\em activity weight} of the state $\rho$.    Due to their relations with the corresponding max relative entropies,  the activity weight and the robustness of activity are non-increasing under arbitrary passivity-preserving operations, and, in particular, under all EPCPR operations.  
 In Appendix \ref{append:as-SDPs-monotone}  we show that the  activity weight and the robustness of activity satisfy the necessary requirements to serve as   {\em bona fide} resource quantifiers \cite{Chitambar2019} and provide upper bounds to the ergotropy.    

\medskip 

{\em Activity witnesses.} A way to experimentally detect a resource is to measure a witness,  that is, an observable that has expectation value in a given interval for all free states, and outside that interval for some non-free states (see e.g. \cite{Horodecki1996, Lewenstein2000, Terhal2000, Guhne2009}).  Here, we define an {\em activity  witness} for a  state $\rho  \in  {\sf St}  (S)\setminus \psh$  as an observable $W\ge  0$ such that  $\langle W \rangle_\rho  >  1$  and $\langle W \rangle_\tau  \le 1$ for every passive state $\tau\in\psh$.
   In Appendix \ref{append:witness-characterization}, we provide a characterization of all possible activity witnesses. Examples   are the observables $W_j  =   j\,  |j\rangle \langle j|$ for $j  \ge 2  $.  Another example is $W_{\rm coh}  =  d \,  |\phi_+\rangle \langle \phi_+|$,  where $|\phi_+\rangle   =  \sum_{i=1}^d  \,  |i\rangle/\sqrt d$ is the canonical maximally coherent state.   

The optimal witness for a given state $\rho$ is the witness $W_{\rm opt}$ satisfying the condition  $\Tr [ W_{\rm opt}  \, \rho]   = \max_W  \Tr[  W \,  \rho]$, where the maximum is over all possible $W\ge 0$ satisfying the condition $ \langle W \rangle_\tau  \le 1$ for every passive state $\tau$.   For the optimal witness, the expectation value can be expressed in terms of the max relative entropy of activity, as follows: 
\begin{align}
\langle W_{\rm opt}  \rangle_\rho    =  2^{R^{\rm act}_{\max}  (\rho)} \, .
\end{align}
This relation, proven in Appendix \ref{append:optimal-witness-and-operational},  shows that the max relative entropy of activity can be estimated from experimental data from the measurement of the optimal activity witness $W_{\rm opt}$.

{\em Activity as a resource for the  approximate generation of maximally coherent states.}   In convex resource theories, the max relative entropy is known to quantify the maximum advantage of non-free states over free states in subchannel discrimination tasks \cite{Takagi2019, Takagi2019b}. 
   For the resource theory of activity, we can  provide  an additional  characterization,   in terms of  the task of transforming a given input  state into the maximally coherent state $|\phi_+\rangle$.   
 In general, we allow the task to be achieved approximately ({\em i.e.} with non-unit fidelity) and probabilistically ({\rm i.e.} with non-unit probability of success) and we take the figure of merit to be the product between the fidelity and the success probability.  For a given quantum operation  $\mc Q$, the advantage of an active state $\rho$ over all  passive states is given by $\min_{\tau \in  \psh}  \langle \phi_+  |    \mc Q  (\rho)   |\phi_+\rangle/   \langle \phi_+  |    \mc Q  (\tau)   |\phi_+\rangle $.   We now consider the maximum advantage over all possible    energy-preserving operations  \cite{chiribella2017optimal}, that is,  all completely positive trace non-increasing maps  achievable by setting up an energy-preserving interaction between the system and an auxiliary system with fully degenerate Hamiltonian, measuring the auxiliary system, and postselecting a subset of the measurement outcomes.
 Notably, we obtain the following relation 
\begin{align}\label{optimaladvantage}
\max_{\mc Q \in  {\sf EPO}}  \min_{\tau \in  \psh}  \frac{\langle \phi_+  |    \mc Q  (\rho)   |\phi_+\rangle}{   \langle \phi_+  |    \mc Q  (\tau)   |\phi_+\rangle}   =  2^{R^{\rm act}_{\max} (\rho)}\, , 
\end{align}
where $ {\sf EPO}$ denotes the set of all energy-preserving operations (see Appendix \ref{append:optimal-witness-and-operational} for the derivation). Eq. (\ref{optimaladvantage})  provides an operational interpretation of the max relative entropy of activity:  $R_{\max}^{\rm act}  (\rho)  $ is the maximum advantage offered by the state $\rho$ in the  probabilistic approximate generation of the maximally coherent state through energy-preserving operations.  

  Interestingly, if the maximization  in Eq. (\ref{optimaladvantage}) is restricted to the set of energy-preserving channels  (energy-preserving operations that are also trace-preserving), we obtain the relation  
  \begin{align}
 \nonumber  \max_{\mc E \in {\sf EPC}}  \min_{\tau \in  \psh}  \frac{\langle \phi_+  |    \mc E  (\rho)   |\phi_+\rangle}{   \langle \phi_+  |    \mc E  (\tau)   |\phi_+\rangle}    &   =d \, \max_{\mc E \in  {\sf EPC}}  \langle \phi_+  |    \mc E  (\rho)   |\phi_+\rangle  \\
&  =   2^{R^{\rm coh}_{\max} (\rho)} \, ,  \label{optimaladvantage1}
  \end{align}  
where $R^{\rm coh}_{\max}  (\rho)   =   \min\{  D_{\max}  (  \rho \|  \gamma)~|~  \gamma \ge 0  \,, \Tr [  \gamma] =1 \, ,  [  \gamma,  H]  =  0 \}$ is the max relative entropy of coherence \cite{Bu2017}, {\em i.e.}
 the minimum of the max relative entropy between $\rho$ and an arbitrary incoherent state  $\gamma$. 
  Here, the  equality in the third line follows from the relations   $\langle \phi_+  | \tau  |\phi_+\rangle = 1/d$ and  $\mc E  (\tau)  =  \tau$, valid for every passive state $\tau$ and for every energy-preserving channel $\mc E$.   Finally, the equality in the second line is proved in Appendix \ref{append:operational-coherence}.   Eq. (\ref{optimaladvantage1})    is also relevant to the resource theory of coherence, as it provides an alternative characterization of the max relative entropy of coherence, in the spirit of Ref. \cite{Bu2017} but different from the results therein.

{\em Discussion.}
In this paper we formulated a resource theory of activity based on  set of EPCPR channels.   It is interesting to consider alternative choices based on larger sets of  operations.   A possible enlargement is the set of UAENICPR  defined earlier in the paper.    Another approach is to follow the lead of the resource theory of coherence and define an analogue of the set of dephasing  covariant operations \cite{Marvian2016,Chitambar2016}, that is, operations $\mc C$ that commute with the complete dephasing map $\Delta (\rho)  :=  \sum_i  \langle i|  \rho  |i\rangle  \,  |i\rangle\langle i|$.   For this purpose, one would need to define the analogue of the complete dephasing in a resource theory of activity.  In Appendix \ref{append:pcos} we show that a perfect analogue is not possible, but some sort of weak analogy is provided by the quantum channel   $\Pi (\rho) :  =  \sum_i \,  \langle i|  \rho  |  i\rangle  \,  \tau_i$, with $\tau_i  =  \sum_{j\le i}  |j  \rangle \langle  j|/i$.   We call $\Pi$ the {\em canonical passivization}. 
 We then define the set of quantum operations that are {\em passivization-covariant},  that is, operations $\mc C$ such that $\mc C \circ \Pi  =  \Pi \circ \mc C$.       In Appendix \ref{append:pcos} we show that all passivization-covariant maps are passivity-preserving, and include the set of  energy-preserving operations as a subset.  Moreover,  we find that Eqs.  (\ref{optimaladvantage}) and (\ref{optimaladvantage1}) still hold when the maximization is extended to the sets of passivization-covariant operations and passivization-covariant channels, respectively.  
 On the other hand, a major drawback of the set of passivization-covariant operations is that they can generally increase the ergotropy: indeed, the example of ergotropy-increasing operation  in Eq. (\ref{counterexample}) is also passivization-covariant. For this reason, passivization-covariant operations do not appear to be an appropriate set of free operations for a resource theory of activity.  
 
We conclude by discussing  the extension  of our results to   Hamiltonians with degenerate spectrum.  In this case, the set of all energy-preserving channels cannot be chosen as the set of free operations, because it includes operations that are not passivity-preserving.  For example, the energy-preserving channel  $\mc E (\rho)   = \langle 1|  \rho  |1\rangle \,  |1\rangle\langle 1|  +  \langle 2|  \rho  |2\rangle \,  |2\rangle \langle 2|    +  \langle 3|  \rho  |3\rangle \,  |2\rangle \langle 2|$ transforms the passive state $\rho =  I/3$ into the active state  $\mc E  (\rho) =  2/3 \,  |2\rangle \langle 2|  + 1/3  |1\rangle \langle 1|$ for a three-dimensional system with degenerate Hamiltonian $H  =  |2\rangle \langle 2|  +  |3\rangle \langle 3|$.   To address this problem, one can consider the subset of  random unitary energy-preserving channels, which are guaranteed to be ergotropy non-increasing by the argument provided earlier in this paper.  Another possible choice is the larger set of unital, energy-preserving channels, which are also ergotropy non-increasing.     A third, intermediate choice is to consider noisy energy-preserving operations, that is, operations generated by setting up an energy-preserving interaction between the system and an auxiliary system with fully degenerate energy levels, initially in the maximally mixed state.    Combined with passive resets, these three sets of operations define three valid candidates for a resource theory of activity in the  degenerate setting.  The most appropriate choice  between these sets is likely to depend on the applications, as random unitary and noisy energy-preserving channels have a clearer physical interpretation, while unital energy-preserving channels have a simpler mathematical characterization.

{\em Conclusions.}  In this paper we introduced a resource theory of activity, providing a resource-theoretic framework for single-system quantum thermodynamics in the absence of heat baths. A natural avenue of future research is to apply our framework to  composite systems where one part is regarded as a work medium and the other part is regarded as a finite heath bath, thereby retrieving a subset of the thermal operations  as effective evolutions of the work medium, and interpolating between the fully bath-less scenario of this paper and the arbitrarily large heath baths considered in earlier studies of quantum thermodynamics.

\begin{acknowledgements}
S. and G.C. acknowledge support from the Hong Kong Research Grant Council through grant 17300918, and through the Senior Research Fellowship Scheme SRFS2021-7S02. 
  U.S. acknowledges  support from Polish National Science Center (NCN) (Grant No. 2019/35/B/ST2/01896). This publication was made possible through the support of the ID$\#$ 62312 grant
from the John Templeton Foundation, as part of the ‘The Quantum Information Structure of Spacetime’
  Project (QISS). The opinions expressed in this project/publication are those of the author(s) and do not necessarily reflect the views of the John Templeton Foundation.   
\end{acknowledgements}

\bibliography{activity}

\appendix

\bigskip

\section{State transitions via EPCPR operations}  
\label{append:state-transformation}

In the non-degenerate case considered in this paper, the set of energy-preserving channels on a $d$-dimensional quantum system  is in one-to-one correspondence with the set of $d\times d$ {\em correlation matrices}, that is,  positive semidefinite matrices $\xi   \ge0 $ satisfying the condition $\xi_{ii}=1$ for every $i\in  \{1,\dots, d\}$.   The correspondence between energy-preserving channels and correlation matrices is established by the following theorem, proven in  Ref. \cite{buscemi2005inverting}: 
\begin{theorem}\label{theo:schur} 
For every energy-preserving channel $\mc E $  there exists   a correlation matrix $\xi$ such that 
\begin{align}\label{schur}
\mc E  (\rho)   =   \xi  \odot \rho  \qquad\forall \rho \in  {\sf St} (S) \, ,
\end{align} 
where $\odot$ denotes the Schur product, defined by the relation $(\xi \odot \rho)_{ij}:  = \xi_{ij}  \,  \rho_{ij}$, $\forall i,j$.      {\em Vice-versa}, every correlation matrix $\xi$ defines an energy-preserving channel $\mc E$ via Eq. (\ref{schur}). 
\end{theorem}

Using Theorem \ref{theo:schur}, we can write a generic EPCPR channel as 
\begin{align}\label{schurEPCPR}
\mc C  (\rho)   =  p  \,   \xi  \odot \rho  +  (1-p ) \,  \tau  \, ,
\end{align}  
  where $p$ is a probability, $\xi$ is an arbitrary correlation matrix, and $\tau$ is an arbitrary passive state.   

Now, suppose that the state transition $\rho  \mapsto  \rho'$ can be achieved by EPCPR channels, namely $\rho'  =  \mc C  (\rho)$ for some EPCPR channel $\mc C$, written  as in Eq. (\ref{schurEPCPR}).  The equality  $ \rho'  = p  \,   \xi  \odot \rho  +  (1-p ) \,  \tau $ is equivalent to  the conditions
\begin{align}
\nonumber \rho'_{mm}  &=  p\,  \rho_{mm} +  (1-p) \,  \tau_{mm}  \qquad & \forall  m\in  \{1,\dots, d\}\\
\rho_{mn}'   &=  p  \,  \xi_{mn} \rho_{mn}  &  \forall m\not =  n  \, .   \label{EPCPRtransitions}
\end{align}

Since the state $\tau$ in Eq. (\ref{schurEPCPR}) is an arbitrary passive state,  the first condition in Eq.  (\ref{EPCPRtransitions}) is equivalent to the condition  
\begin{align}
\rho_{mm}'   -  p  \, \rho_{mm} \ge \rho_{nn}'   -  p  \, \rho_{nn} 
\qquad \forall m\le n\, ,
\end{align}
which in turn is equivalent to 
\begin{align}
\rho_{mm}' - \rho_{nn}'  \ge  p\,  ( \rho_{mm}   -  \rho_{nn})  \qquad  \forall m\le n \,,
\end{align}
and to 
\begin{align}\label{almostthere}
\rho_{mm}' - \rho_{m+1, m+1}'  \ge  p\,  ( \rho_{mm}   -  \rho_{m+1, m+1})   \qquad \forall m \, .
\end{align}
Another way to express this condition is to define the sets   ${\sf  S}_+   :=  \{   m~|~      \rho_{mm} - \rho_{m+1, m+1}  \ge 0  \}$ and ${\sf  S}_-   :=  \{   m~|~    \rho_{mm} - \rho_{m+1, m+1}  < 0  \}$.   With these definitions, the existence of a  $p\in  [0,1]$ satisfying  Eq. (\ref{almostthere}) is equivalent to existence of a $p\in  [0,1]$ satisfying the condition 
\begin{align}
p_-  \le p  \le p_+
\end{align}
with
\begin{align}
\nonumber &p_{+}  :=   \min_{m\in   {\sf S}_+}  \frac{  \rho'_{mm}  -  \rho'_{m+1 ,  m+1}  }{ \rho_{mm}  -  \rho_{m+1 ,  m+1}  }\\  
\nonumber &p_{-}:  =  \max_{m\in   {\sf S}_-}  \frac{  \rho'_{m+1, m+1}  -  \rho'_{m ,  m}  }{ \rho_{m+1,  m+1}  -  \rho_{m ,  m}  } \, 
\end{align} 
(here we use the convention  $(\rho'_{mm}  -  \rho'_{m+1 ,  m+1})/(  \rho_{mm}  -  \rho_{m+1 ,  m+1}  )  =    {\sf sign}  (\rho'_{mm}  -  \rho'_{m+1 ,  m+1})  \times   \infty$   if $ \rho_{mm}  -  \rho_{m+1 ,  m+1}=0$, where ${\sf sign} (x)  =  x/|x|$.) 

Now, let us suppose that   $\rho_{mn}  \not  =  0$ for every $m$ and $n$  in  $\{1,\dots,  d\}$. In this case, the second condition in Eq. (\ref{EPCPRtransitions})   is equivalent to $\rho'_{mn}/\rho_{mn}  =  p  \xi_{mn}$ for every $m\not =  n$.    Since the matrix $\xi$ in Eq. (\ref{schurEPCPR})  is an arbitrary correlation matrix, the condition  $\rho'_{mn}/\rho_{mn}  =  p  \xi_{mn}  \, \forall m\not  =  n$   is equivalent to nonnegativity of the matrix $p\,  \xi    =  R  +  p\,  I$, with  
\begin{align}\label{ratiomatrix}  R_{mn} :  =  \begin{cases}  \frac{\rho_{mn}'}{ \rho_{mn}}  \qquad &   m\not  =  n  \, ,  \rho_{mn} \not  =  0\\
0    &  m  =  n   \, .
\end{cases}
\end{align}
Since $\rho_{mn}\not =  0$ for every $m,n$, the above equation completely specifies the matrix $R$. Note that the matrix $R$ is Hermitian, and therefore diagonalizable.  
 Hence, nonnegativity of the matrix $R  +  p\,  I$ is equivalent to the condition $\min {\rm eigv}  (R)  \ge  -  p$, where $\min {\rm eigv}  (R)   :=  \min_{|\psi\rangle  : \,  \| |\psi\>\|  =1}  \,  \langle\psi | \xi  |\psi\rangle $ is the minimum eigenvalue of $R$.   
In summary,  Eqs. (\ref{EPCPRtransitions}) are equivalent to the conditions 
\begin{align}
\nonumber &p_+  \ge p_- \\
 \nonumber  &[p_-,  p_+]  \cup  [0,1]  \not  = \emptyset\\
  \label{convertibilityconditions}  &\min {\rm eigv}  (R)   \ge -  \min\{p_{+},  1\} \, .
\end{align}
These conditions completely characterize the state transitions achievable  by EPCPR in the case $\rho_{mn}\not= 0\, , \forall m,n$.   

Suppose that some of the off-diagonal elements $\rho_{mn}$, $m\not =  n$, are zero.  In this case,     Eq. (\ref{ratiomatrix})   does not completely determine the matrix $ R$.  The transition $\rho  \mapsto \rho'$ can be achieved by EPCPR if and only if  the conditions Eq. (\ref{convertibilityconditions}) are satisfied for some Hermitian matrix $R$ satisfying  Eq. (\ref{ratiomatrix}).  For small $d$, the existence of the matrix $R$ can be tested numerically, by randomly picking the missing matrix elements $R_{mn}=  |R_{mn}| e^{i\theta_{mn}}$ (for $m$ and $n$ such that $\rho_{mn} = 0$) with an arbitrary modulus  $|R_{mn}|\in [0,1]$ and an arbitrary phase $\theta_{mn}  \in  [0,2\pi]$.

\section{Properties of the activity weight and robustness of activity}
\label{append:as-SDPs-monotone}

Here we derive some useful properties of the activity weight and robustness of activity.

\subsection{Semidefinite programs}

We start by providing semidefinite programs for the activity weight and robustness of activity, respectively.  The semidefinite programs are expressed in terms of a fixed (but otherwise arbitrary)  channel   $\mc P$  satisfying the condition $\mc P  ({\sf St}  (S))  =   \psh$, that is, a quantum channel that (1)  maps all  states into  passive states and  (2)  is surjective over the set of passive states. The characterization of such channels is provided in the end of this subsection.

\begin{proposition}[SDP for the  activity weight]
\label{prop:weightSDP}
For a given state $\rho$, the activity weight $\aw$ can be equivalently written as 
\begin{align}
\label{eq:aw-sdp-appen}
\aw =1-\max_{\substack{\Gamma\geq 0 \\ \mc{P}(\Gamma) \leq \rho}} \Tr[\Gamma] \, ,
\end{align}
where $\mc P$ is any arbitrary channel satisfying the condition $\mc P  ({\sf St}  (S))  =   \psh$, 
and 
\begin{align}
\label{eq:aw-sdp-appen1}
\aw  =  1-\min_{\substack{\xi\geq 0 \\ \mc{P}^{\dagger}(\xi) \geq \mathbb{I}}} \Tr[\xi\rho] \, ,
\end{align}
where $\mc P^\dag$ is the adjoint of $\mc P$,  uniquely defined by the equation $\Tr  [   A   \mc P  (B)  ]  =  \Tr[  \mc P^\dag  (A) \,   B]$ for arbitrary  $d\times d$ matrices $A$ and $B$.
\end{proposition}

\begin{proof}
The definition of activity weight implies the alternative expression  
\begin{align}
\label{eq-append:acw1}
\aw&=\min_{\tau\in\psh}\{t\geq 0:\rho\geq (1-t)\, \tau\},
\end{align}  
following  from the fact that the operator inequality $\rho \ge   (1-t) \,\tau$ is equivalent to the equality $\rho  =  (1-t) \tau +  t \eta$ for a suitable density matrix $\eta$. 

We then have  the following chain of equalities:
\begin{align*}
\aw &=\min_{\tau\in\psh}\{t\geq 0:\rho\geq (1-t)\, \tau\}\\
&=\min_{\sigma\in{\sf St} (S)}\{t\geq 0:\rho\geq (1-t)\mc{P}(\sigma)\}\\
&=\min_{\Gamma  \geq 0}\{1- \Tr[\Gamma]:\rho\geq \mc{P}(\Gamma)\}\\
&=1-\max_{\substack{\Gamma \geq 0\\ \mc{P}(\Gamma) \leq \rho}} \Tr[\Gamma].
\end{align*}
where the  second equality follows  from the condition  $\mc P  ({\sf St}  (S))  =   \psh$,  and the third equality follows from defining $\Gamma   :=  (1-t)\, \sigma$. We have thus proved  Eq. (\ref{eq:aw-sdp-appen}). 
Eq. (\ref{eq:aw-sdp-appen1})  then follows from the duality of semidefinite programming, which in this particular problem yields an equality between the solutions of the primal and dual problems. 
\end{proof}

\medskip  

The above proposition has a simple consequence:
\begin{corollary}
For a given state $\rho$, the activity weight $\aw$ can be equivalently written as 
\begin{align}\label{weightSDPconceptual}
\aw  =  1-\min_{\substack{\xi\geq 0 \\   \Tr [\xi \tau]   \ge 1\, ~   \forall \tau  \in  \psh}} \Tr[\xi\rho] \, ,
\end{align}
\end{corollary}

\begin{proof}  
The condition  $\mc P^\dag (\xi)  \ge I$ in Eq.  (\ref{eq:aw-sdp-appen1}) is equivalent to the condition $\Tr  [  \sigma  \,  \mc P^\dag  (\xi)]  \ge  \Tr [\sigma] =1$ for every quantum state $\sigma \in  {\sf St  }(S)$.  
 In turn, this condition is equivalent to $\Tr [  \mc P( \sigma) \,  \xi  ]   \ge 1$, $\forall \sigma \in  {\sf St  }(S) $.  Since $\mc P    ( {\sf St  }(S))  = \psh$, the last condition is equivalent to $\Tr  [  \tau  \,  \xi]  \ge 1$ for every passive state $\tau$.      
\end{proof}
\medskip 

We now provide analogous expressions for the robustness of activity.

\begin{proposition}[SDP for the  robustness of activity]
\label{prop:robustnessSDP}
For a given state $\rho$, the robustness of activity  $\roa \rho$ can be equivalently written as 
 \begin{align}
\label{eq:roa-sdp} \roa \rho  =\min_{\substack{\Gamma\geq 0 \\ \mc{P}(\Gamma) \geq \rho}} \Tr[\Gamma]   -1  \, ,
\end{align}
where $\mc P$ is an arbitrary channel satisfying the condition $\mc P  ({\sf St}  (S))  =   \psh$, 
and 
\begin{align}
\label{eq:roa-sdp1}
\roa \rho=\max_{\substack{\xi\geq 0 \\ \mc{P}^{\dagger}(\xi) \leq \mathbb{I}}} \Tr[\rho\xi]  -1 \, .
\end{align}
\end{proposition}

\begin{proof}
The definition of robustness of activity implies the alternative expression  
\begin{align}
\label{eq-append:roa2}
\roa{\rho}=\min_{\tau\in\psh}\{t\geq 0:\rho\leq (1+t) \, \tau\}.
\end{align}
 We then have  the following chain of equalities:
\begin{align*}
\roa{\rho} &=\min_{\tau\in\psh}\{t\geq 0:\rho\leq (1+t) \, \tau\}\\
&=\min_{\sigma \in {\sf St} (S) }\{t\geq 0:\rho\leq (1+t)\mc{P}(\sigma)\}\\
&=\min_{\Gamma \geq 0 }\{\Tr[\Gamma]-1:\rho\leq \mc{P}(\Gamma)\}\\
&=\min_{\substack{\Gamma \geq 0\\ \mc{P}(\Gamma)\geq \rho }}\Tr[\Gamma]-1 \, 
\end{align*}
where the  second equality follows  from the condition  $\mc P  ({\sf St}  (S))  =   \psh$,  and the third equality follows from defining $\Gamma   :=  (1+t)\, \sigma$.
This proves  Eq. (\ref{eq:roa-sdp}). 
Again, Eq. (\ref{eq:roa-sdp1})  follows from the duality of semidefinite programming. 
\end{proof}

\medskip

\begin{corollary}
For a given state $\rho$, the robustness of activity $\roa \rho$ can be equivalently written as 
\begin{align}\label{robustnessSDPconceptual1}
\roa \rho  = \max_{\substack{\xi\geq 0 \\   \Tr [\xi \tau]   \le 1\, ~   \forall \tau  \in  \psh}} \Tr[\xi\rho]    -1\, ,
\end{align}
Similarly, the max relative entropy of activity is given by  
\begin{align}\label{robustnessSDPconceptual}
R_{\max}^{\rm act}  (\rho)   = \log  \max_{\substack{\xi\geq 0 \\   \Tr [\xi \tau]   \le 1\, ~   \forall \tau  \in  \psh}} \Tr[\xi\rho]   \, ,
\end{align}
\end{corollary}

\begin{proof}  
The condition  $\mc P^\dag (\xi)  \le I$ in Eq.  (\ref{eq:roa-sdp1}) is equivalent to the condition $\Tr  [  \sigma  \,  \mc P^\dag  (\xi)]  \le  \Tr [\sigma] =1$ for every quantum state $\sigma \in  {\sf St  }(S)$.  
 In turn, this condition is equivalent to $\Tr [  \mc P( \sigma) \,  \xi  ]   \le 1$, $\forall \sigma \in  {\sf St  }(S) $.  Since $\mc P    ( {\sf St  }(S))  = \psh$, the last condition is equivalent to $\Tr  [  \tau  \,  \xi]  \le 1$ for every passive state $\tau$.   Eq.  (\ref{robustnessSDPconceptual}) then follows from the relation    $R_{\max}^{\rm act}  (\rho)   =    \log \left[  \roa \rho  + 1\right]$.      
\end{proof}

\medskip 

The above corollary shows that the max relative entropy of activity of a given state $\rho$ is equal to the expectation value of the optimal activity witness for the state $\rho$.

We conclude this section by characterizing all the quantum  channels $\mc P$ satisfying the conditions $\mc P  ({\sf St}  (S))  =  \psh$. 
To this purpose, we start from the following proposition:  
\begin{proposition}\label{prop:entbreak}
Suppose that the Hamiltonian of system $S$ has non-degenerate spectrum.  Then,  a quantum  channel  $\mc P$ on system $S$ is activity breaking if and only if it is  of the form  
\begin{align}\label{measprep}
\mc P  (\rho)  =   \sum_{j=1}^d  \, \Tr[  P_j\,  \rho]\,  \tau_j \, , 
\end{align}  
where $(P_j)_{j=1}^d$ is a positive operator-valued measure (POVM)  (that is,  $P_j\ge 0  \, ,\forall j$ and $\sum_j  P_j  =  I$)  and
\begin{align}\label{tauj}
\tau_j  : =  \frac{  \sum_{i=1}^j \,  |i\rangle\langle i|}j \, .
\end{align}
   \end{proposition}
\begin{proof}
Since the states $\{\tau_j\}_{j=1}^d$ are passive, it is immediate to see that every channel of the form (\ref{measprep}) is activity breaking.  We now prove the converse.  
In the nondegenerate case, the convex set of all  passive states is a simplex, with the states $\{  \tau_j~|~  i\in\{1,\dots,  d\}\}$ as the extreme points. Hence, the state  $\mc P  (\rho)$ must be a convex combination of these states for every activity breaking channel $\mc P$ and for every state $\rho$:   in formula, let us write  
\begin{align}\label{convmix}\mc P  (\rho)   =  \sum_j  \,  p_j  (\rho)  \,  \tau_j\, ,
\end{align} where each $p_j(\rho)$ is a probability. 

{\em A priori,} the probability $p_j (\rho)$ could have any dependence on $\rho$.  In our case, however, the dependence is linear, as we show in the following.  Consider  the operators 
\begin{align}
\nonumber &A_j   :  =  j    (  |j\rangle \langle j|   -   |j+1\rangle \langle j+1|  ) \qquad j\in  \{1,\dots,  d-1\}\\
&A_d: =  d  \,  |d\rangle\langle d|  
\end{align} 
and note that they satisfy  the condition  
\begin{align}
\Tr [  A_j\,   \tau_k ]  =  \delta_{jk}  \qquad \forall j,k  \, . 
\end{align}
Hence, Eq. (\ref{convmix}) implies 
\begin{align}
p_j (\rho)   =   \Tr [A_j\,  \mc P  (\rho)]  =  \Tr  [  P_j \,   \rho ]   \,,
\end{align}
with $P_j:  = \mc P^\dag  (A_j)$.   Since $p_j(\rho)$ is positive for every state $\rho$, we must have $P_j \ge 0$. Moreover, the normalization condition  $\sum_j   p_j (\rho)  = 1$ for every $\rho$ implies $\sum_j  P_j  =  I$.  Hence, the operators $(P_j)_{j=1}^d$ form a POVM.
\end{proof}

\medskip  

Proposition \ref{prop:entbreak} implies that every activity breaking channel must be entanglement-breaking (in the case of nondegenerate Hamiltonians.)    We now restrict our attention to activity breaking channels that are  surjective on the set of passive states.

\begin{proposition}\label{prop:surjectiveactbreak}
Suppose that the Hamiltonian of system $S$ has non-degenerate spectrum. Then, a  quantum channel   $\mc P$  on system $S$ satisfies  the condition $\mc P ({\sf St}  (S))   =  \psh$  if and only if it is  of the form 
\begin{align}\label{vnmeasprep}
\mc P (\rho)  =  \sum_{j=1}^d \,  \langle \psi_j|  \,\rho  \,  |\psi_j\rangle \,  \tau_j\, ,
\end{align} 
where the vectors $\{  |\psi_j\rangle\}_{j=1}^d$ form an orthonormal basis.  
\end{proposition}

\begin{proof}  It is immediate to see that every quantum channel of the form (\ref{vnmeasprep}) satisfies the condition  $\mc P ({\sf St}  (S))   =  \psh$.  Now, we prove the converse implication. First, note that the condition    $\mc P ({\sf St}  (S))   \subseteq  \psh$ implies that $\mc P$ must be of the entanglement-breaking form (\ref{measprep}). Now, the surjectivity condition $\mc P ({\sf St}  (S))   =  \psh$ implies that there exist $d$ states $\{\rho_j\}_{j=1}^d$ such that 
\begin{align}
\tau_j  =  \mc P  (\rho_j)     =  \sum_k \Tr [   P_k \, \rho_j]\,  \tau_k  \qquad\forall j\in  \{1,\dots,  d\} \,.  
\end{align}
Since the states $\{\tau_j\}_{j=1}^d$ are linearly independent, the above equality implies the condition 
$\Tr  [  P_k \, \rho_{j}]= \delta_{j,k}$ for every $j$ and $k$. Hence, the states $\{\rho_{j}\}_{j=1}^d$  are perfectly distinguishable, which means that they have orthogonal supports. But for a quantum system of dimension $d$,  any $d$ states with orthogonal support must necessarily be pure. Hence, there exists an orthonormal basis $\{|\psi_j\rangle\}_{j=1}^d$ such that $\rho_j  =  |\psi_j\rangle\langle \psi_j|  =  P_j$.   
\end{proof}

\medskip  

An interesting technical result, which will be used later in the paper, is the following characterization of the set of correlation matrices: 
\begin{proposition}\label{characterizationcorrelation}
A positive matrix $\xi$ is a correlation matrix if and only   $\mc P^\dag  (\xi)   =   I$ for some channel $\mc P$  satisfying the condition $\mc P  ({\sf St}  (S))  =   \psh$.  
\end{proposition} 

\begin{proof}
Let $\mc P$ be an arbitrary channel satisfying the condition $\mc P  ({\sf St}  (S))  =   \psh$.   If $\xi$ is a correlation matrix,  then   $\mc P^\dag (\xi)   =  \sum_i  \, \Tr [  \xi  \, \tau_i]\,  |  \psi_i\rangle \langle \psi_i|     =   \sum_i  \, |  \psi_i\rangle \langle \psi_i|  =  I$, where we used Eq. (\ref{vnmeasprep}) and the condition $\Tr [\xi  \tau_i]=1 \, ,\forall i$ valid for every correlation matrix $\xi$. 

Conversely, suppose that $\mc P^\dag  (\xi)  =  I$ for some channel $\mc P$  satisfying the condition $\mc P  ({\sf St}  (S))  =   \psh$.   Then,  we must have $\mc P^\dag (\xi)   =  \sum_i  \, \Tr [  \xi  \, \tau_i]\,  |  \psi_i\rangle \langle \psi_i|      =  I  =     \sum_i  \, |  \psi_i\rangle \langle \psi_i| $, which implies $\Tr [  \xi  \, \tau_i]  =1$ for every $i$.  In turn, this condition implies $\xi_{ii}=1$ for every $i$. Hence, $\xi$ is a correlation matrix. 
\end{proof}

\subsection{Faithfulness, convexity, and monotonicity under passivity-preserving operations}  

We now show that the activity weight and robustness of activity are two {\em bona fide} measures of non-passivity: they are faithful, convex, and non-increasing under arbitrary passivity-preserving channels.

\begin{proposition}[Properties of the activity weight]
\label{claim:a-w-mon}
For every system $S$, the activity weight ${\sf A}_w$ is 
\begin{enumerate}
\item contained in the interval  $[0,1]$
\item  {\em faithful}, {\em i.e.} $\aw =0$ if and only if   $\rho$ is passive. 
\item  {\em convex}, {\em i.e.} $\awe{p \rho_1+(1-p)\rho_2}\leq p \awe{\rho_1}+(1-p)\awe{\rho_2}$, for every probability $p\in [0,1]$ and every pair of states $\rho_1,\rho_2\in {\sf St} (S)$.
\item  {\em non-increasing under passivity-preserving channels}, {\em i.e.}    $\awe{\msc{C}(\rho)}\leq \aw$ for every  passivity-preserving channel  $\msc{C}$ and every state $\rho$. 
\end{enumerate}
\end{proposition}

\begin{proof}
Containment in the interval $[0,1]$ and faithfulness follow directly from the definition of the activity weight.

Let us now prove  convexity.  For every pair of  states $\rho_1$ and $\rho_2$ and for every probability $p_1$, we have 
\begin{align}
\nonumber \awe  {p\, \rho_1  +   (1-p) \, \rho_2 }   &=1-\max_{\substack{\Gamma\geq 0 \\ \mc{P}(\Gamma) \leq p\rho_1  +  (1-p  \rho_2)}} \Tr[\Gamma]   \\
\nonumber &  \le 1  -  p  \,  \max_{\substack{\Gamma_1\geq 0 \\ \mc{P}(\Gamma_1) \leq \rho_1  }} \Tr[\Gamma_1  ]  \\
 \nonumber  &  \qquad  -  (1-p)    \,  \max_{\substack{\Gamma_2\geq 0 \\ \mc{P}(\Gamma_2) \leq \rho_2  }} \Tr[\Gamma_2  ] \\
  &     = p   \,   \awe   {\rho_1}  +  (1-p) \,  \awe{  \rho_2 }  \, ,
\end{align}
where the first and last equality follow from  Eq. (\ref{eq:aw-sdp-appen}). 

For the proof of monotonicity under passivity-preserving channels,  consider an optimal decomposition such that $\rho = \aw\eta^* +(1-\aw)\tau^*$, where  $\tau^*$ is some passive state and $\eta^*$ is an arbitrary state. Then, we have $\msc{C}(\rho) = \aw \msc{C}  (\eta^*   ) +(1-\aw) \widetilde{\tau}$, 
where $\widetilde{\tau}:=\msc{C}  (\tau^*)$ is a passive state  since  $\msc{C}$ is a passivity-preserving channel. Thus, $ \awe{\msc{C}(\rho)}\leq \aw$. 
\end{proof}

\medskip 

We now provide the analogous result  for the robustness of activity.

\begin{proposition}[Properties of the robustness of activity]
\label{claim:roa-mon}
For every system $S$, the robustness of activity ${\sf A}_r$ is 
\begin{enumerate}
\item contained in the interval $[0,d-1]$.
\item  {\em faithful}, {\em i.e.} $\roa \rho=0$ if and only if   $\rho$ is passive. 
\item  {\em convex}, {\em i.e.} $\roa{p \rho_1+(1-p)\rho_2}\leq p \roa{\rho_1}+(1-p)\roa{\rho_2}$, for every probability $p\in [0,1]$ and every pair of states $\rho_1,\rho_2\in {\sf St} (S)$.
\item  {\em non-increasing under passivity-preserving channels}, {\em i.e.}    $\roa{\msc{C}(\rho)}\leq \roa{\rho}$ for every  passivity-preserving channel  $\msc{C}$ and every state $\rho$. 
\end{enumerate}
The maximum value of the robustness of activity is $d-1$.  
\end{proposition}

\begin{proof}
Nonnegativity and faithfulness follow immediately from the definition.    For the  upper bound  $\roa \rho  \le d-1$,   note  that every density matrix $\rho$ satisfies the operator inequality $\rho\leq \mathbb{I}=(1+t)\frac{\mathbb{I}}{d}$ with $t=  d-1$. Hence, Eq. \eqref{eq-append:roa2} implies $\roa {\rho} \le d-1$.  The upper bound is achieved by the maximally active state $\rho  =  |  d\rangle \langle d|$.  

We now prove convexity.  For every pair of  states $\rho_1$ and $\rho_2$ and for every probability $p_1$, we have 
\begin{align}
\nonumber \roa  {p\, \rho_1  +   (1-p) \, \rho_2 }   &=\min_{\substack{\Gamma\geq 0 \\ \mc{P}(\Gamma) \geq p\rho_1  +  (1-p  \rho_2)}} \Tr[\Gamma]    -1  \\
\nonumber &  \le    p  \,  \min_{\substack{\Gamma_1\geq 0 \\ \mc{P}(\Gamma_1) \geq \rho_1  }} \Tr[\Gamma_1  ]  \\
 \nonumber  &  \qquad  + (1-p)    \,  \min_{\substack{\Gamma_2\geq 0 \\ \mc{P}(\Gamma_2) \geq \rho_2  }} \Tr[\Gamma_2  ]  -1 \\
  &     = p   \,   \roa   {\rho_1}  +  (1-p) \,  \roa{  \rho_2 }  \, ,
\end{align}
where the first and last equality follow from  Eq. (\ref{eq:roa-sdp}).

To prove monotonicity,  consider an optimal decomposition  $\tau^* =  \frac{\rho+ \roa{\rho}\eta^*}{1+\roa{\rho}}$, where  $\tau^*$ is a passive state and $\eta^*$ is a suitable state. Then, we have   $\widetilde  \tau =  \frac{\msc C(\rho)+ \roa{\rho} \,  \msc  C(\eta^*) }{1+\roa{\rho}}$,
where $\widetilde{\tau}:=\msc{C}   (\tau^*)$ is a passive state, since  $\msc{C}$ is passivity-preserving. Hence, we obtained the inequality $ \roa{\msc{C}(\rho)}\leq \roa{\rho}$. \end{proof}

\subsection{Relations between activity weight, robustness of activity, and ergotropy}

The activity weight and the robustness of activity are related by an elementary inequality:  
\begin{proposition}
For every quantum  state $\rho  \in  {\sf St}  (S)$, one has  the bound
\begin{align}
\aw  \geq \frac{\roa{\rho}}{d-1}   \, .\label{aw-ar}
\end{align}
\end{proposition}
\begin{proof}Consider the optimal decomposition $\rho=\aw \eta^*+(1-\aw)\tau^*$ for $\rho$, with $\tau^* \in \psh$. Since  $\roa{\rho}$ is convex in $\rho$,  we have
\begin{align*}
\roa{\rho}
&\leq \aw\roa{\eta^*}+ (1-\aw)\roa{\tau^*}\\
&=\aw\roa{\eta^*}\\
& \leq   (d-1)  \, \aw .
\end{align*}
where we have used the fact that $\roa{\tau^*}=0$ in the second line, and the upper bound  $\roa{\eta^*}\leq d-1$ in the third line. 
\end{proof}

\medskip

We now provide two upper bounds on the ergotropy in terms of the activity weight and robustness of activity, respectively. 
Note that in general the ergotropy  ${\sf Erg}(\rho)$ is bounded as 
\begin{align}
       \langle H \rangle_\rho   -   \frac{E_1  +  E_2  + \dots  +  E_d  }d  \le     {\sf Erg}(\rho)  \le  \langle H \rangle_\rho   -   E_1   
  \, 
\end{align}
for a Hamiltonian with nondegenerate eigenvalues $E_1 <  E_2  <   \cdots  <  E_d$.        The lower  bound follows from the fact that the maximally mixed state $I/d$ is the passive state with maximum average energy.

\begin{proposition}
For every quantum  state $\rho  \in  {\sf St}  (S)$, one has  the upper  bounds
\begin{align}
\label{eq:lem-work}
   {\sf Erg}  (\rho) \leq     \aw  \,        \max_{\sigma :  {\rm Supp}  (\sigma) \subseteq {\rm Supp}  (\rho)}  {\sf Erg}(\sigma)     
  \,  ,\end{align}
and
\begin{align}\label{eq:lem-work-low}
{\sf Erg}(\rho)  \le    \min\{\roa \rho,  1\}   \,         (E_{i_{\max}}  -  E_1 )\, , 
\end{align}   
with $i_{\max} :  = \max\{  i~|~  \langle i|  \rho |i \rangle \not  =  0    \}$. 
\end{proposition}

  \begin{proof}   
 Let $\rho  =  t \,  \eta+   (1-t)  \, \tau  $ be a convex decomposition where  $\tau$ is some passive state and $\eta$ is a suitable state.   Since the ergotropy is a convex function, we have 
  \begin{align}
  \nonumber {\sf Erg}(\rho)  &  \le     t \,    {\sf Erg}  (\eta)  +   (1-t)  \,    {\sf Erg} (\tau)  \\
    & =   t\,  {\sf Erg} (\eta)  \, .      
  \end{align}
Picking the decomposition with $t=  \aw$ and using the relation  ${\sf Erg}  (\eta)  \le \max_{\sigma:   {\rm Supp}  (\sigma) \subseteq {\rm Supp}  (\rho)}  \,    {\sf Erg}(\sigma)$ we then obtain Eq. (\ref{eq:lem-work}).

To prove Eq. (\ref{eq:lem-work-low}), consider a convex decomposition $\tau   = \frac {  \rho  +  t\,  \eta}{1+t}$,  where  $\tau$ is some passive state and $\eta$ is a suitable state.    
    Let  $U_*$ be a unitary operator such that ${\sf Erg}  ( \rho)  =  \langle  H\rangle_\rho    - \langle  H\rangle_{  U_*\rho  U_*^\dag}$.         Denoting by $P$ the projector on ${\sf Span}  \{  |1\rangle,  |2\rangle,  \dots,  |i_{\max} \rangle  \}$, we  must have that $U_*  P  U_*^\dag  =  P$ in order  for $\langle  H\rangle_{  U_*\rho  U_*^\dag}$ to be minimum.    Moreover, we have 
  \begin{align}
    \langle  H\rangle_\tau   -     \langle  H\rangle_{U_*  \tau  U_*^\dag}    &   =  \frac{ {\sf Erg}(\rho)     +  t\,     \left[     \langle  H\rangle_\eta   -     \langle  H\rangle_{U_*  \eta  U_*^\dag}  \right]}{1+t} \, .
  \end{align}
  Since $\tau$ is passive, the l.h.s. is negative.  Hence, we have  the inequality 
  \begin{align}
  {\sf Erg}(\rho)  \le  t  \,     \left[       \langle  H\rangle_{U_*  \eta  U_*^\dag}   -   \langle  H\rangle_\eta \right] \, .   
  \end{align}
Picking a decomposition   $\tau_*   = \frac {  \rho  +  t_*\,  \eta_*}{1+t_*}$ for which $t_*=  \roa \rho$, we then obtain  
  \begin{align}
  {\sf Erg}(\rho)  \le   \roa \rho \,     \left[       \langle  H\rangle_{U_*  \eta_*  U_*^\dag}   -   \langle  H\rangle_{\eta_*} \right] \, .   
  \end{align}
Now,   using the relation $P\rho  P=  \rho$,   we obtain  the relation  $P  \tau_*  P =  \frac {  \rho  +  t_*\,  P\eta_* P}{1+t_*}$, or equivalently, $\tau'  =  \frac {  \rho  +  t'\,  \eta' }{1+t'} $ where $\tau':  =  P\tau_* P/\Tr[\tau_*  \,P]$ is a passive state,  $\eta'  :=  P  \eta_*  P/\Tr[  \eta_*\,  P] $, and $t'  =  t_*  \,  \Tr [ \eta_*  \,  P]$.    Since $t_*   =   \roa \rho$  is the minimum over all possible convex decompositions of the form $\tau   = \frac {  \rho  +  t\,  \eta}{1+t}$, with passive $\tau$, we conclude that $t'   \ge  t_* $ and therefore $\Tr  [\eta_* \,  P] \ge  1$ (or equivalently, $\Tr  [\eta_* \,  P] = 1$).  Hence, $\eta_*$ must have support contained in the support of $P$, and the same must hold for $U_* \eta_*  U_*^\dag$, since $U_*$ leaves the support of $P$ invariant.  Hence, we have $ \langle  H\rangle_{U_*  \eta_*  U_*^\dag}  $,  and therefore  $   \langle  H\rangle_{U_*  \eta_*  U_*^\dag}   -   \langle  H\rangle_{\eta_*}  \le  E_{i_{\max}}   -  E_{1} $.    

Summarizing, we obtained the inequality  $  {\sf Erg}(\rho)  \le   \roa \rho \,   ( E_{i_{\max}}  -  E_1)$, which combined with the trivial inequality $ {\sf Erg}(\rho)  \le   E_{i_{\max}}  - E_1$ yields Eq. (\ref{eq:lem-work-low}). \end{proof}

\section{Activity witnesses} 
\label{append:witness-characterization}

Here we provide a characterization of the set of activity witnesses.  

\begin{proposition}
For every positive semidefinite matrix $W$, the following conditions are equivalent: 
\begin{enumerate}
\item  $\Tr [   W   \,  \tau] \le 1$ for every passive state $\tau$.
\item   The diagonal matrix elements of  $W$ satisfy the conditions 
\begin{align}
\nonumber W_{11}  &\le 1  \\
\nonumber W_{11}  +  W_{22}  &  \le 2  \\
\nonumber  & \vdots \\
W_{11} +  \cdots +  W_{dd}  & \le d  \label{witnesses}
 \, .  
 \end{align}
\end{enumerate}\end{proposition}

\begin{proof} The condition   $\Tr[ W\, \tau]  \le 1 \,  ~\forall \tau  \in  \psh$ is equivalent to  
\begin{align}\label{witness}
\Tr [W \,  \tau_j]    \le 1   \qquad \forall  j\in  \{1,\dots, d\} \, ,
\end{align} 
where $\tau_j  =  \sum_{i\le j}  |i\rangle \langle  i|/j$ are the extreme points of $\psh$.     Inserting the explicit expression of the states $\tau_j$ into Eq. (\ref{witness}) we then obtain Eq. (\ref{witnesses}).  
\end{proof}

\medskip  

We say that a positive semidefinite matrix $W$ is a {\em nontrivial} activity witness if $\Tr [W  \,  \tau  ]\le 1$ for every passive state $\tau$ and $\Tr[  W\,\rho]  >1$ for some state $\rho$.  
The above characterization provides a  way to construct  nontrivial activity witnesses.   For example, the matrices 
\begin{align}
W_k=  \,  k\, |  k\rangle \langle k|     
\end{align}
are nontrivial activity witnesses for every $k\ge 2$, as the matrix $W_k$ witnesses the activity of the state $\rho_k  =   |  k\rangle \langle k| $.       Another example of nontrivial activity witness is $W  =  d\,  |\phi_+ \rangle \langle \phi_+|$, which witnesses the activity of the maximally coherent state $|\phi_+\rangle$.

In the following, we will use a connection between activity witnesses and {\em sub-correlation matrices}, that is, positive semi-definite matrices $\xi\ge 0 $ satisfying the condition $\xi_{ii} \le 1$ for every $i\in  \{1,\dots, d\}$.  The connection is established in the following lemma: 
\begin{lemma}\label{lem:witnessescorrelation}
Every sub-correlation matrix $\xi  \ge0 $ is an activity witness.  {\em Vice-versa}, every activity witness $W$ is proportional to a sub-correlation matrix:  specifically, the matrix $W/d$ is a sub-correlation matrix for every activity witness $W$.   
\end{lemma}
 
 \begin{proof}
 Clearly, the condition $\xi_{ii}\le 1, ~\forall i$ implies the condition $\sum_{i=1}^j  \xi_{ii}  \le j, ~\forall  j$.    On the other hand, since the matrix elements $W_{ii}$ are nonnegative, the condition $\sum_{i=1}^j  W_{ii} \le j \, ,\forall  j$ implies in particular $W_{jj}  \le  j  \le d,   \, \forall j\in \{1,\dots, d\}$.   Hence, $W/d$ is a sub-correlation matrix.      
  \end{proof}

\section{Operational interpretation of the max relative entropy of activity}
\label{append:optimal-witness-and-operational}

Consider the operational scenario in which an experimenter can set up an energy-preserving interaction between the system $S$ and an auxiliary system with fully degenerate energy levels, measure the auxiliary system, and postselect a subset of the outcomes.   The resulting quantum operations, called {\em energy-preserving} \cite{chiribella2017optimal}, can be characterized as follows:  
\begin{proposition}  
For a quantum system with nondegenerate Hamiltonian $H$, the following are equivalent
\begin{enumerate}
\item The quantum operation $\mc Q$ is energy-preserving.
\item There exists an energy-preserving channel $\mc C$ such that the map $\mc C  -  \mc Q$ is completely positive.
\item   $\mc Q$ has a Kraus representation of the form $\mc Q  (\rho)   =  \sum_{j}  Q_j \rho  Q_j^\dag$, with $[Q_j,  H]  =  0  $ for every $j$.
\item The action of $\mc Q$ on an input state $\rho$ is given by $\mc Q  (\rho)  =   \xi  \odot  \rho$, where $\xi$ is a sub-correlation matrix, that is, a positive matrix satisfying the condition  $\xi_{ii}    \le 1$ for every $i \in \{1,\dots, d\}$.  
\end{enumerate}
\end{proposition}

\begin{proof}
$1  \Rightarrow 2$.  By definition, an energy-preserving quantum operation $\mc Q$ is contained in an energy-preserving instrument  \cite{chiribella2017optimal}, that is, a set of  quantum operations  $(\mc Q_x)_{x\in\sf X}$ that sum up to an energy-preserving channel  $\mc C  :  =  \sum_x  \, \mc Q_x$.   Hence,  one  $\mc Q  =  \mc Q_{x_0}$ for some $x_0  \in  \sf X$. Now, $\mc C  :  =  \sum_x  \, \mc Q_x$ is an energy-preserving channel and $\mc C -  \mc Q  = \sum_{x\not =  x_0}\,\mc Q_x$ is a completely positive map.  

$2  \Rightarrow 3$.    Since the map $\mc C -\mc Q$ is completely positive, it has a Kraus representation, say $(\mc C  - \mc Q)  (\rho)  =  \sum_k  \,   R_k \rho  R_k^\dag$.  Then, one has the Kraus representation  $\mc C  (\rho)    =  \sum_j   Q_j \, \rho  \, Q_j^\dag  +  \sum_k  R_k\,  \rho  \,  R_k^\dag$, where   $   \mc Q (\rho)  =  \sum_j \,  Q_j  \rho  Q_j^\dag$ is a Kraus representation for $\mc Q$.  Now,  Ref. \cite{chiribella2017optimal} showed that every Kraus representation of an energy-preserving channel $\mc C$ is of the form $\mc C  (\rho)  =  \sum_l\,  C_l  \rho  C_l^\dag$ with $[C_l,H]=0$ for every $l$. Applying this condition to the above Kraus representation, we obtain in particular $[Q_j,  H ]=0$.  

$3 \Rightarrow 4$.    Let $   \mc Q (\rho)  =  \sum_j \,  Q_j  \rho  Q_j^\dag$ be a Kraus representation with $[Q_j,  H]=   0$ for every $j$.  Writing $Q_j  = \sum_{i}   \alpha_{ij}\,  |i\rangle \langle i|$ for suitable coefficients $\alpha_{ij}$, we then obtain  $\mc Q  (\rho)  =  \sum_{i,i',j}  \alpha_{ij}  \, \overline \alpha_{i'j} \,   |i\rangle \langle i|  \rho  |i'\rangle \langle i'|  =  \sum_{i,i'}   \xi_{ii'} \rho_{ii'} |i\rangle \langle i'| \,  $, with $\xi_{ii'}   := \sum_{j}  \,  \alpha_{ij}  \overline \alpha_{i'j}$.   By construction, the matrix $\xi$ is positive semidefinite and $\xi_{ii}   =  \Tr[  \mc Q  (|i\rangle \langle i|)] \le 1$ for every $i$.    
 
$4 \Rightarrow 1$.   Suppose that $\mc Q  (\rho)   =  \xi \odot \rho$ for some sub-correlation matrix $\xi$.  
 Let $\xi'$ be the diagonal matrix with $\xi_{ii}'  :  =  1-\xi_{ii}$.  Then,   $\xi  +  \xi'$ is a correlation matrix and the map $\mc E$ defined by $\mc E  (\rho ):  =   (\xi  +  \xi')  \odot  \rho$ is an energy-preserving channel by Theorem \ref{theo:schur}.     Hence, the quantum operations $\mc Q$ and $\mc Q'  :  \rho  \mapsto  \mc Q'  (\rho)  =  \xi'\odot \rho$ form an energy-preserving instrument, in the sense of Ref. \cite{chiribella2017optimal}. There, it was shown that every energy-preserving instrument can be realized by coupling the system with an auxiliary system with fully degenerate energy levels, and by performing measurements on the auxiliary system.   Hence, the quantum operation $\mc Q$ is energy-preserving.  
\end{proof}

Using the above characterization we can evaluate the maximum advantage of a state $\rho$ in the task of generating the canonical maximally coherent state through energy-preserving operations: 
\begin{align}
\nonumber &\max_{\mc Q  \in  {\sf EPO}}   \min_{\tau  \in \psh}  \frac{  \langle \phi_+  |  \mc Q  (\rho) |\phi_+ \rangle}{  \langle \phi_+  |  \mc Q  (\tau) |\phi_+ \rangle}   \\
  \nonumber   & \qquad    =  \max_{\xi \ge 0,  \xi_{ii}  \le 1 \forall i}    \min_{\tau \in  \psh}   ~\frac{ \langle  \phi_+  |    \xi  \odot \rho \,  |\phi_+
\rangle}{\langle  \phi_+  |    \xi  \odot \tau \,  |\phi_+
\rangle}\\
 \nonumber   & \qquad    =  \max_{\xi \ge 0,  \xi_{ii}  \le 1 \forall i}    \min_{\tau \in  \psh}   ~\frac{ \Tr[    ( \xi \odot   |\phi_+
\rangle  \langle  \phi_+  |) \,  \rho ] }{ \Tr[    ( \xi \odot   |\phi_+
\rangle  \langle  \phi_+  |) \,  \tau ] }\\
 & \qquad    =  \max_{\xi \ge 0,  \xi_{ii}  \le 1 \forall i}    \min_{\tau \in  \psh}   ~\frac{ \Tr[    \xi  \,  \rho ] }{ \Tr[      \xi  \,  \tau ] } \, ,\label{fff}
\end{align}
the last equality following from the condition  $   \xi \odot   |\phi_+
\rangle  \langle  \phi_+  |    =  \xi/d$.

Now, Lemma \ref{lem:witnessescorrelation} implies that the maximum  over all sub-correlation matrices can be replaced  by a maximum over all activity witnesses. Hence, we obtain     
\begin{align}
\nonumber &\max_{\mc Q  \in  {\sf EPO}}   \min_{\tau  \in \psh}  \frac{  \langle \phi_+  |  \mc Q  (\rho) |\phi_+ \rangle}{  \langle \phi_+  |  \mc Q  (\tau) |\phi_+ \rangle}   \\
\nonumber &  \qquad   =  \max_{\xi \ge 0,  \xi_{ii}  \le 1 \forall i}    \min_{\tau \in  \psh}   ~\frac{ \Tr[    \xi  \,  \rho ] }{ \Tr[      \xi  \,  \tau ] } \\
 \nonumber   & \qquad    =  \max_{W \ge 0,    \Tr [W  \tau] \le 1  ~ \forall   \tau  \in \psh}    \min_{\tau \in  \psh}   ~\frac{ \Tr[    W  \,  \rho ] }{ \Tr[      W  \,  \tau ] } \\
 \nonumber   & \qquad    =  \max_{W \ge 0,    \Tr [W  \tau] \le 1  ~ \forall   \tau  \in \psh}      ~ \Tr[    W  \,  \rho ] \\
 &  \qquad =  2^{R_{\max}^{\rm act}  (\rho)} \, ,  \label{ggg}
\end{align}
where the last equality follows from Eq. (\ref{robustnessSDPconceptual}), and the third equality follows by restricting without loss of generality the maximization to the {\em normalized} activity witnesses, that is, the activity witnesses satisfying $\Tr[ W  \tau_*] =1$ for some passive state $\tau_*$.

\section{Max relative entropy of coherence and energy-preserving channels}
\label{append:operational-coherence}

Here we provide an operational characterization of the max relative entropy of coherence 
\begin{align}
R_{\max}^{\rm coh}  (\rho)  :  =  \min  \{  D_{\max}  (\rho \|  \gamma)~|~  \gamma \ge 0  ,  \Tr[\gamma]=1 \,  [  \gamma,  H]   =  0 \} \, .
 \end{align}  
Ref. \cite{Bu2017} showed the relation   
\begin{align}\label{set}
2^{R_{\max}^{\rm coh}  (\rho)}    =  \max_{\mc C \in  \sf Set}   d\,  \langle \phi_+  |   \mc C (\rho)  |\phi_+\rangle \, ,  
\end{align} 
where the maximization is over a set of  quantum channels $\sf Set$, which can be either the set of all incoherent channels, the set of strictly incoherent channels, or the set of dephasing covariant channels.   Here, we show that Eq.  (\ref{set})  remains true even if the maximization is restricted to the strictly smaller set of energy-preserving channels.  
The proof is simple:  using the characterization of the energy-preserving channels (Theorem \ref{theo:schur}), we obtain 
\begin{align}
 \nonumber \max_{\mc E \in  \sf  EPC}   d\,  \langle \phi_+  |   \mc E (\rho)  |\phi_+\rangle   
 &  =  \max_{\xi \ge 0  \, , \xi_{ii}  \ge 1 \forall i }   d\,  \langle \phi_+  |    \xi  \odot \rho    |\phi_+\rangle   \\
\nonumber  &  =  \max_{\xi \ge 0  \, , \xi_{ii} =  1 \forall i }     \Tr  [ d\,( \xi \odot  |\phi_+ \rangle \langle \phi_+  |   )    \,\rho ]  \\
   &  = \max_{\xi  \ge 0 \, ,\xi_{ii}  = 1  \, \forall i} \Tr  [ \xi \, \rho]    \, ,  
\end{align}
where the third equality follows from the relation $\xi \odot  |\phi_+ \rangle \langle \phi_+  |     = \xi/d $.

Now, the maximization over all correlation matrices $\xi$ can be equivalently expressed as  
\begin{align}\label{aaa}
\max_{\xi  \ge 0 \, ,\xi_{ii}  = 1  \, \forall i} \Tr  [ \xi \, \rho]  & =  \max_{\xi \ge 0 \, , \Delta  (\xi) =   I}  \,  \Tr  [  \xi \rho] \, ,
\end{align}
where $\Delta$ is the completely dephasing channel $\Delta (\rho):=  \sum_{i}  \langle i|  \rho  |i\rangle  \,  |i\rangle \langle i|  $.   At this point, the duality of semidefinite programming implies the bound 
\begin{align}\label{bbb}
\nonumber  \max_{\xi \ge 0 \, , \Delta  (\xi) =   I}  \,  \Tr  [  \xi \rho]     &   \le \min_{\Delta  (\sigma)   \ge   \rho,\, \sigma \ge 0 }  \, \Tr  [\sigma]   \\
 \nonumber   &   =  \min_{\Gamma  \ge   \rho, \,   [  \Gamma,  H]  =  0   } \, \Tr [\Gamma]\\
 \nonumber   &   =  \min_{t\,  \gamma  \ge   \rho, \,   [  \gamma,  H]  =  0 \, , \Tr[\gamma]=1  } \, t \\
   & =   2^{ R_{\max}^{\rm coh}  (\rho) } \,,
  \end{align}
 where the second line comes from the fact that  $\Gamma  := \Delta  (\sigma)$ is an arbitrary diagonal matrix,  or equivalently, an arbitrary matrix satisfying $[ \Gamma,  H]  =  0$. 
 
The bound is achieved with the equality sign because the above semidefinite program satisfies the condition for strong duality \cite{Watrous_lecture_notes_2018}.

Summarizing, we have proven the equality 
\begin{align}
  \max_{\mc E \in  \sf  EPC}   d\,  \langle \phi_+  |   \mc E (\rho)  |\phi_+\rangle     =   2^{ R_{\max}^{\rm coh}  (\rho) }   \, .
\end{align}

\section{Passivization-covariant operations}
\label{append:pcos}

An approach to construct a resource theory is to pick a prototype of resource-destroying operation, and then to define the free operations as those that commute with it. 
 For example, a  resource theory of coherence can be built from the complete dephasing channel $\Delta:  \rho \mapsto  \Delta (\rho)  =  \sum_i   \langle i| \rho  |i\rangle   \,  |i\rangle \langle i|$,  by taking the free operations to be  {\em dephasing covariant operations}  \cite{Marvian2016,Chitambar2016}, that is, the operations $\mc C$ such that $\mc C  \circ \Delta  = \Delta \circ \mc C$.   In the case of activity, following this scheme would require fixing a prototype of activity-breaking channel.

 \subsection{The  passivization channel}
Ideally, the prototype of activity-breaking channel should leave all passive states invariant.  
Unfortunately, however, this condition cannot be satisfied: 
 \begin{proposition}\label{prop:nocanonicalpass}
 No activity-breaking map  $\mc P$ can satisfy the condition $\mc P  (\tau)= \tau$ for every passive state $\tau$.  
 \end{proposition} 
 \begin{proof}
 Suppose that the map $\mc P$ satisfies the condition  $\mc P  (\tau)= \tau, \, \forall \tau \in \psh$. In particular, this condition must apply to the extreme points of $\psh$: one must have  $\mc P (\tau_i) =  \tau_i$ for every $i\in \{1,\dots,  d\}$.  On the other hand, Eq. (\ref{tauj}) shows that every extremal point $\tau_i$ with $i  \ge 2$ can be decomposed as $\tau_i  =  \frac{i-1}i   \, \tau_{i-1}  +     \frac 1i  \,   |i\rangle\langle i|$.   Then, the condition  $\mc P (\tau_i) =  \tau_i$ for every $i\in  \{1,\dots, d\}$ implies  $\mc P (|i\rangle \langle i| )  = |i\rangle \langle i|$ for every $i    \in  \{1,\dots, d\}$. Since $|i\rangle \langle i|$ is non-passive for every $i \ge 2$, the map $\mc P$ cannot be activity breaking. 
  \end{proof}

\medskip 

In fact, an even stronger no-go result holds: no activity-breaking channel can preserve the set of passive states. 
\begin{proposition}
 For every activity-breaking channel $\mc P$,  the inclusion  $\mc P  (\psh)\subset \psh$ is strict.  In other words, there is no activity breaking channel $\mc P$ such that $\mc P  (\psh)= \psh$.  
 \end{proposition} 
 \begin{proof} The proof is by contradiction. Suppose that there existed a channel  $\mc P$ such that $\mc P  (\psh)= \psh$.    Then, there must exist some passive states $(\rho_i)_{i=1}^d$ such that $\mc P(\rho_i)   =  \tau_i$ for every $i$.     Since each $\tau_i$ is an extreme point of $\psh$ and $\mc P$ is surjective on $\psh$, each $\rho_i$ must be an extreme point of $\psh$.    Hence, there must exist a permutation $\pi$ such that $  \rho_i  =  \tau_{\pi(i)}$.   In fact, the permutation must be the identity. Indeed, the  inclusion relation  ${\sf Supp}  (\rho_i)  \subseteq  {\sf Supp}  (\rho_j)$ implies ${\sf Supp}   (\mc P  (\rho_i))   \subseteq {\sf Supp}   (\mc P  (\rho_j))$ for all  $i$ and $j$ such that $i\le j$.   Hence, the relation   
 \begin{align}{\sf Supp}  (\tau_1)  \subseteq  {\sf Supp}  (\tau_2)  \subseteq  \cdots \subseteq  {\sf Supp}  (\tau_d) 
 \end{align}  
 implies     the condition
 \begin{align}{\sf Supp}  (\tau_{\pi^{-1} (1)})  \subseteq  {\sf Supp}  (\tau_{\pi^{-1} (2)})  \subseteq  \cdots \subseteq  {\sf Supp}  (\tau_{\pi^{-1}  (d)}) \, ,
 \end{align}  
which is satisfied only if $\pi$ is the identity permutation.  Hence,  the channel $\mc P$ must leave all the extreme points of $\psh$ invariant, which in turn means that $\mc P$ leaves the whole set $\psh$ invariant.    Then, Proposition \ref{prop:nocanonicalpass} implies that $\mc P$ cannot be activity-breaking
   \end{proof}
  
\medskip

As a weaker requirement, one can consider activity-breaking channels that are surjective on the set of passive states, {\em i.e.} the set of quantum channels $\mc P$ such that $\mc P  ({\sf St}  (S))   =   \psh$.  The possible candidates are characterized  in Proposition \ref{prop:surjectiveactbreak}: 
all such channels are of the form  
\begin{align}\label{Psurjective}
\mc P (\rho)   =  \sum_{i=1}^d  \langle \psi_i|  \rho  |\psi_i\rangle\,  \tau_i \, ,
\end{align}
for some orthonormal basis $\{|\psi_i\rangle\}_{i=1}^d$.  

We now select a privileged channel of the form (\ref{Psurjective}), by imposing some additional requirements.  First, we require  the channel $\mc P$ to commute with the time evolution generated by the system's Hamiltonian, namely  
\begin{align}\label{timecov}
\mc U_t  \circ \mc P  =  \mc P \circ \mc U_t\qquad \forall t\in  {\mathbb R} \,,
\end{align}
 where $\mc U_t$ is the unitary channel associated to the unitary operator $U_ t  :=  \exp  [- i   t  H]$. In other words, we require the channel $\mc P$ to be covariant with respect to the time evolution generated by the system's Hamiltonian.

\begin{proposition}
A  quantum channel  $\mc P$ of the form (\ref{Psurjective}) satisfies the covariance condition  (\ref{timecov})  if and only if it is of the form 
\begin{align}
\label{def:act-break}
  \mc{P}(\rho) = \sum_{i=0}^{d-1}\Tr[\rho \Op{i}] \, \tau_{\pi (i)},
\end{align}
where $\pi$ is a permutation of the set $\{1,\dots,  d\}$.
\end{proposition} 
\begin{proof}   It is immediate to see that every map of the form (\ref{def:act-break}) is of the form (\ref{Psurjective}) and is covariant with respect to time evolution.  To prove the converse, let $\mc P$ be  a channel of the form (\ref{Psurjective}).    Since $\mc P$ is activity-breaking, its output is invariant under time evolution. Hence, the covariance condition (\ref{timecov}) becomes $  \mc P \circ \mc U_t  = \mc P$, $\forall t \in \mathbb R$, or explicitly
\begin{align}
\nonumber 
&\sum_i  \langle \psi_i|  \,  \mc U_t  (  \rho)\,    |\psi_i\rangle~   \tau_i  =  \sum_i  \langle \psi_i|  \,   \rho\,    |\psi_i\rangle~   \tau_i   \\
& \qquad \qquad \qquad \qquad \qquad  \forall \rho  \in  {\sf St}  (S)  \, ,\forall t\in  \mathbb R\, .
\end{align}
    Since the states $\tau_i$ are linearly independent, this condition implies  
    \begin{align}
 \nonumber &   \langle \psi_i|  \,  \mc U_t  (  \rho)\,    |\psi_i\rangle  = \langle \psi_i|  \,   \rho\,    |\psi_i\rangle \\
 &  \qquad \quad \qquad   \qquad \forall  i\in  \{1,\dots, d\}\, , \rho  \in  {\sf St}  (S) \,,\forall t\in\mathbb R
    \end{align}
    or equivalently  
    \begin{align}
    \mc U_t^\dag  (|\psi_i\rangle \langle \psi_i|)   =  |\psi_i\rangle \langle \psi_i|  \qquad \forall   i\in  \{1,\dots, d\} \,,\forall t\in\mathbb R
   \, .
   \end{align}
   Hence, the states $\{|\psi_i\rangle\}_{i=1}^d$ should be the eigenstates of the Hamiltonian $H$, that is $|\psi_i\rangle  =    |\sigma  (i) \rangle$ for some permutation $\sigma $. Defining $\pi   =  \sigma^{-1}$ one then has  Eq.   (\ref{def:act-break}). 
\end{proof}
 
 \medskip 
To conclude, we require  that the canonical passivization channel $\mc P$ should respect the ordering of the energy eigenstates, 
 namely 
\begin{align}
\Tr  [ H\, \mc P( |i\rangle \langle i|)   ]  \le \Tr [ H \,\mc P(|j\rangle \langle j| )]  \qquad \forall i\le j \, .
\end{align}  
  This condition requires the permutation  $\pi$ in Eq.  (\ref{def:act-break}) to be the identity permutation. Summarizing, we have shown the following theorem: 
  \begin{theorem}
  For a quantum system with non-degenerate  Hamiltonian $H   =  \sum_{i=1}^d  \,  E_i \,  |i\rangle \langle i|$, the quantum channel $\Pi$ defined by 
  \begin{align}
  \Pi (\rho)   =  \sum_{i=1}^d \, \langle i|  \rho  |i\rangle \,  \tau_i   
  \end{align} is the only channel that 
   \begin{enumerate}
  \item is activity-breaking,
  \item  is surjective on the set of passive states,
  \item  is covariant with respect to time evolution generated by the Hamiltonian $H$, and
  \item preserves the energy ordering of the energy eigenstates.
  \end{enumerate}
  \end{theorem}
  In the following, we will call $\Pi$ the {\em canonical passivization}.

\subsection{Passivization-covariant operations}

We now define the set of {\em passivization-covariant operations (PCOs)} as the set of quantum operations $\mc Q$ that commute with the canonical passivization $\Pi$, namely
\begin{align}
    \mc Q \circ \Pi = \Pi \circ \mc Q  \, .
\end{align}
When the PCO $\mc Q$ is  trace-preserving, we call it a {\em passivization-covariant channel}.  

\begin{proposition}
Every passivization covariant operation  is  passivity-preserving. 
\end{proposition} 

\begin{proof}  Let $\mc Q$ be an arbitrary passivization-covariant operation.  We now show that $\mc Q(\tau)$ is proportional to a passive state whenever $\tau$ is a passive state.   Since the canonical passivization $\Pi$ is surjective on the set of passive states, every  passive state $\tau\in \psh$ can be written as  $\tau  =  \Pi  (\rho)$ for some suitable state $\rho  \in {\sf St}  (S)$.  Then, one has 
\begin{align}
\mc Q (\tau)  =  \mc Q  \circ \Pi  (\rho)  =  \Pi  \circ   \mc Q (\rho)  \, .
\end{align}
Since $\Pi$ is activity breaking, we conclude that $\mc Q (\tau) $ is proportional to a passive state.
\end{proof}

\medskip

\begin{proposition}
Every energy-preserving channel is  passivization-covariant.
\end{proposition} 

\begin{proof}  For every energy-preserving channel $\mc E$ and for every state $\rho  \in {\sf St}  (S)$, one has
\begin{align}
\nonumber \Pi   \circ \mc E   (\rho)    &   =  \sum_i   \langle i|  \mc E (\rho)  |i\rangle  \,  \tau_i  \\
&  =     \sum_i   \langle i| \rho  |i\rangle  \,  \tau_i    \, ,
\end{align}
the second equality following from the energy-preserving condition.  On the other hand,   one has
\begin{align}\mc E \circ  \Pi  (\rho)   =  \sum_i   \langle i|  \rho  |i\rangle  \,  \mc E (\tau_i)  =  \sum_i   \langle i|  \rho  |i\rangle  \, \tau_i  \, .
\end{align}
Hence, we obtained  the relation $\Pi   \circ \mc E   (\rho)  =   \mc E \circ  \Pi  (\rho)  \, , \forall \rho \in {\sf St}  (S)$, that is, $\Pi   \circ \mc E  =  \mc E  \circ\Pi$.    
\end{proof}

\medskip  

\begin{proposition}\label{prop:channel}
For every passivization-covariant operation $\mc Q$ there exists a passivization-covariant channel $\mc C$ such  that the map $\mc C -  \mc Q$ is completely positive.  
\end{proposition}  

\begin{proof}
Define  $Q  :  =  \mc Q^\dag  (I)$.  Since the passivization map $\Pi$ is trace-preserving, we have $\Pi^\dag (I)  =  I$, and therefore  $Q   =   \mc Q^\dag  (  \Pi^\dag (I) ) =  \Pi^\dag  \,  (\mc Q^\dag (I))  =   \Pi^\dag  (Q)$.       Now, define the quantum operation $\mc Q' $ as $\mc Q'   (\rho)  = (\Tr [\rho]-   \Tr [Q \rho]   )\,  |1\rangle \langle 1|$. Note that $\mc Q'$ is passivization-covariant, since one has  
\begin{align}
\nonumber \Pi  \circ \mc Q'  (\rho)     &= \left(  \Tr  [\rho] -  \tr  [  Q \rho] \right) \,  \Pi  (|1\rangle \langle 1|)   \\
\nonumber &= \left(\Tr [\rho]-   \tr  [  Q \rho]\right)  \, |1\rangle \langle 1| \\
\nonumber &=   \left(\Tr [\rho]-   \tr  [  \Pi^\dag (Q) \rho] \right) \, |1\rangle \langle 1| 
\\
\nonumber &  = \left(  \Tr  [  \Pi (\rho)] -  \tr  [ Q  \Pi( \rho)] \right) \, |1\rangle \langle 1| \\
&  =  \mc Q'  \circ \Pi  (\rho) 
\end{align}
for every quantum state $\rho$.  Moreover, $\mc C:  = \mc Q  +  \mc Q'$  is trace-preserving, since one has $\Tr  [  \mc C  (\rho)]   =   \Tr  [  \mc Q   (\rho) ]  +  \Tr  [ \mc Q'  (\rho)]  =   \Tr  [Q  \rho ]  + \left( \Tr  [\rho]  -\Tr[Q  \rho]  \right)   =  \Tr  [\rho]$ for every $\rho$.  Hence, $\mc C$   is a passivization-covariant channel.  
\end{proof}

\medskip

A notable property of passivization-covariant channels  is provided by the following lemma: 
\begin{lemma}\label{lem:ppccorrelation}
For every passivization-covariant channel $\mc C$, the matrix $   d\, \mc C^\dag  (|\phi_+\rangle\langle  \phi_+|)$  is a correlation matrix.  
\end{lemma}

\begin{proof}  The proof uses Proposition \ref{characterizationcorrelation}, which guarantees that the matrix  $  \xi: =   d\, \mc C^\dag  (|\phi_+\rangle\langle  \phi_+|)$  is a correlation matrix  if   $\Pi^\dag  (\xi)=  I$.   This condition is indeed satisfied:  
\begin{align}
\nonumber \Pi^\dag  (\xi)   &=  d\,   \Pi^\dag    \circ \mc C^\dag  (|\phi_+\rangle\langle  \phi_+|)  \\
\nonumber &=  d\,       \mc C^\dag \circ \Pi^\dag  (|\phi_+\rangle\langle  \phi_+|)\\
\nonumber &  =  \mc C^\dag  (  I) \\
 &=  I\, ,
 \end{align} 
where the fourth equality follows from the fact that $\mc C$ is trace-preserving (and therefore $\mc C^\dag (I)  = I$) and the third equality follows from the condition $\Pi^\dag (|\phi_+\rangle\langle  \phi_+|)  = \sum_i  \,  \langle \phi_+|  \tau_i|\phi_+\rangle   \,  |i\rangle\langle i|  =  \sum_i  |i\rangle\langle i|/d = I/d $.  
\end{proof}

\medskip

Using the above property, we can derive an alternative operational interpretation of the max relative entropy of activity and the max relative entropy of coherence: \\

\begin{proposition}
The following equalities hold: 
\begin{align}\label{ccc}
2^{R_{\max}^{\rm coh}  (\rho)}   =  \max_{\mc C \in  {\sf PCC}}   \min_{\tau  \in  \psh} \frac{   \langle \phi_+|   \mc C  (\rho) | \phi_+\rangle}{ \langle \phi_+|   \mc C  (\tau) | \phi_+\rangle}\, ,
\end{align}
and
\begin{align}\label{ddd}
2^{R_{\max}^{\rm act}  (\rho)}   =  \max_{\mc Q \in  {\sf PCO}}   \min_{\tau  \in  \psh} \frac{   \langle \phi_+|   \mc Q  (\rho) |\phi_+\rangle}{ \langle \phi_+|   \mc Q  (\tau) | \phi_+\rangle}
\end{align} 
where the maxima are over the set of  all passivization-covariant channels ($\sf PCC$) and over the set of all all passivization-covariant quantum operations ($\sf PCO$), respectively.
\end{proposition}

\begin{proof}
Let us start from the proof of Eq. (\ref{ccc}).  We have 
\begin{align}
\nonumber  \max_{\mc C \in  {\sf PCC}}   \min_{\tau  \in  \psh} \frac{   \langle \phi_+|   \mc C  (\rho) | \phi_+\rangle}{ \langle \phi_+|   \mc C  (\tau) | \phi_+\rangle}  &  =  \max_{\mc C \in  {\sf PCC}}   \min_{\tau  \in  \psh} \frac{ \Tr [  \mc C^\dag  ( | \phi_+\rangle \langle \phi_+|)   \rho ] }  { \Tr [  \mc C^\dag  ( | \phi_+\rangle \langle \phi_+|)   \tau ] }\\
\nonumber  &  =  \max_{\xi\ge 0 \,  ,  \xi_{ii}  = 1 \, \forall i}  \min_{\tau  \in  \psh} \frac{ \Tr [  \xi    \rho ] }  { \Tr [  \xi    \tau ] }  \\
\nonumber  &  =  \max_{\xi\ge 0 \,  ,  \xi_{ii}  = 1 \, \forall i}  \Tr [  \xi    \rho ]   \, ,
\end{align}
where the second equality follows from Lemma \ref{lem:ppccorrelation} and the third equality follows from the relation $\Tr  [\xi  \tau]  =1$ valid for every correlation matrix $\xi$ and every passive state $\tau$.   Using Eqs. (\ref{aaa}) and (\ref{bbb}) we then conclude   $ \max_{\xi\ge 0 \,  ,  \xi_{ii}  = 1 \, \forall i}  \Tr [  \xi    \rho ]  =  2^{R_{\max}^{\rm coh}  (\rho)}$.  

To prove Eq.  (\ref{ddd}), we observe that
\begin{align}
\nonumber  \max_{\mc Q \in  {\sf PCO}}   \min_{\tau  \in  \psh} \frac{   \langle \phi_+|   \mc Q  (\rho) | \phi_+\rangle}{ \langle \phi_+|   \mc Q  (\tau) | \phi_+\rangle}  &  =  \max_{\mc Q \in  {\sf PCO}}   \min_{\tau  \in  \psh} \frac{ \Tr [  \mc Q^\dag  ( | \phi_+\rangle \langle \phi_+|)   \rho ] }  { \Tr [  \mc Q^\dag  ( | \phi_+\rangle \langle \phi_+|)   \tau ] }\\
 &  =  \max_{\xi\ge 0 \,  ,  \xi_{ii}  \le 1 \, \forall i}  \min_{\tau  \in  \psh} \frac{ \Tr [  \xi    \rho ] }  { \Tr [  \xi    \tau ] }\, .
\end{align}
The second equation follows from  Proposition \ref{prop:channel}, which guarantees that there exists a passivization-covariant channel $\mc C$ such that $\mc C-  \mc Q$ is completely positive. Hence, one has $d\,  \mc Q^\dag   (  |\phi_+\rangle\langle \phi_+|)   \le d\,   \mc C^\dag  (I) $.  Since $d\,   \mc C^\dag  (I)$  is a correlation matrix (Lemma \ref{lem:ppccorrelation}),  $d\,  \mc Q^\dag   (  |\phi_+\rangle\langle \phi_+|)$ is a sub-correlation matrix.   Eqs.  (\ref{fff}) and (\ref{ggg}) then conclude the proof.
\end{proof}



\end{document}